%% file: paper.tex
\documentclass[conference]{IEEEtran}

\usepackage{tikz}
\usepackage{url}
\usepackage{amsmath,empheq}
\usepackage{multirow}
\usepackage[ruled,vlined]{algorithm2e}
\usepackage{enumitem}
\usepackage{pifont}
\usepackage[listings,skins,breakable]{tcolorbox}
\usepackage{filecontents}
\usepackage{subcaption}
\usepackage{mathdots}
\usepackage{boxedminipage}
\usepackage{amssymb}
\usepackage{amsthm}
\usepackage{comment}
\usepackage{xcolor}
\usepackage{multicol,booktabs, colortbl}
\definecolor{aliceblue}{rgb}{0.94, 0.97, 1.0}
\newcommand{\squishlist}{
	\begin{list}{$\bullet$}
		{
			\setlength{\itemsep}{0pt}
			\setlength{\parsep}{3pt}
			\setlength{\topsep}{3pt}
			\setlength{\partopsep}{0pt}
			\setlength{\leftmargin}{1em}
			\setlength{\labelwidth}{0.8em}
			\setlength{\labelsep}{0.3em} } }
				
\newcommand{\squishend}{
\end{list}  }
\newcommand{\point}[1]{\vspace{0.8ex}\par\noindent{\bf #1.}}

\usepackage{titling}
\newcommand{\tool}{ShadowNet\xspace}

\newcommand\numberthis{\addtocounter{equation}{1}\tag{\theequation}}
\newcommand{\F}{\mathbb{F}}

\newtheorem{theorem}{Theorem}
\newtheorem{lemma}{Lemma}
\newtheorem{conj}{Conjecture}

\newenvironment{sketch}{\begin{proof}[Proof Sketch]}{\end{proof}}
\begin{document}
%-------------------------------------------------------------------------------

%don't want date printed
\date{}

\title{\Large \bf \tool: A Secure and Efficient  On-device
\\    Model Inference System for Convolutional Neural Networks}

\author{Zhichuang Sun\IEEEauthorrefmark{1}, Ruimin Sun\IEEEauthorrefmark{2}, Changming Liu\IEEEauthorrefmark{3}, Amrita Roy Chowdhury\IEEEauthorrefmark{4}, \\Long Lu\IEEEauthorrefmark{3}, Somesh Jha\IEEEauthorrefmark{5}
\\ Google\IEEEauthorrefmark{1}, Florida International University\IEEEauthorrefmark{2}, Northeastern University\IEEEauthorrefmark{3}, \\University of California, San Diego\IEEEauthorrefmark{4}, University of Wisconsin-Madison\IEEEauthorrefmark{5}
}
 
%for single author (just remove % characters)
%\author{
%{\rm Your N.\ Here}\\
%Your Institution
%\and
%{\rm Second Name}\\
%Second Institution
% copy the following lines to add more authors
% \and
% {\rm Name}\\
%Name Institution
%} % end author

\maketitle
%\thispagestyle{plain}
%\pagestyle{plain}

\input{abs}
\input{intro}

\input{background}

\input{overview}
\input{design}
\input{impl}
%\input{analysis}
%\input{shadownet_design}
%\input{observation}
%\input{design}
%\input{demotext}
\input{evaluation}
\input{discussion}
\input{relatedwork}

\input{conclusion}

\section*{Acknowledgement}
The authors would like to thank Matthew Jagielski for the insightful feedback on our security analysis of black-box model extraction attack and the anonymous reviewers for their constructive
comments. This work was supported by the National Science Foundation under Grant No. 2031390. 
%-------------------------------------------------------------------------------
\bibliographystyle{plain}
\bibliography{citations,relatedwork}
\clearpage
\input{appendix}

%\input{RevisionLetter}

%%%%%%%%%%%%%%%%%%%%%%%%%%%%%%%%%%%%%%%%%%%%%%%%%%%%%%%%%%%%%%%%%%%%%%%%%%%%%%%%
\end{document}

%% file: abs.tex
%-------------------------------------------------------------------------------
\begin{abstract}
%--------------------------------------------------
    With the increased usage of AI accelerators on mobile and edge devices, on-device machine learning (ML) is gaining popularity. Thousands of proprietary ML models are being deployed today on billions of untrusted devices.  This raises serious security concerns about model privacy.
    %Encrypting the models is the main deployed protection, yet found vulnerable to unsophisticated runtime attacks. 
    %TEE-based solutions face several challenges in practice. 
    %Simply moving model inference into mobile TEE will exhaust the limited TEE resources,
    %greatly increase the TCB size, and lose access to the untrusted hardware accelerators including GPU. 
    However, protecting model privacy without losing access to the untrusted AI accelerators
    is a challenging problem.
    In this paper, we present a novel on-device model inference system, ShadowNet. ShadowNet
 protects the model privacy with Trusted Execution Environment (TEE)
    while securely outsourcing the heavy linear layers of the model
    to the untrusted hardware accelerators. ShadowNet achieves this by transforming 
    the weights of the linear layers before outsourcing them and restoring the results inside the TEE.
    The non-linear layers are also kept secure inside the TEE.
    %With this split execution, ShadowNet keeps the TEE size small.
    %that protects the model privacy without losing access to the hardware accelerators 
    %ShadowNet uses the untrusted hardware for acceleration, including GPU.
    %ShadowNet restores the result and runs the nonlinear layers securely inside the TEE.
    ShadowNet's design ensures efficient transformation of the weights and the subsequent restoration of the results. We build a ShadowNet prototype based on
    TensorFlow Lite and evaluate it on five popular CNNs, namely, MobileNet, ResNet-44, MiniVGG, ResNet-404, and YOLOv4-tiny.
    Our evaluation shows that ShadowNet achieves strong security guarantees 
    with reasonable performance, offering a practical solution for secure on-device model inference.
\end{abstract}

%% file: intro.tex
%-------------------------------------------------------------------------------
\section{Introduction}
On-device machine learning is becoming increasingly popular as more and more AI accelerators are being used
in mobile and embedded devices, such as NPU~\cite{DBLP:journals/corr/abs-2112-02439}, GPU~\cite{mobilegpu} and Edge TPU~\cite{edgetpu}. 
A recent study~\cite{sun2020mind} has shown that thousands of mobile apps 
are using on-device machine learning (ML) for diverse applications, such as OCR~\cite{ocrapp}, 
face recognition~\cite{faceapp}, liveness detection~\cite{livenessapp}, ID card and bank card recognition~\cite{cardapp} and translation~\cite{translateapp}.
The benefits of on-device machine learning are obvious; it avoids sending user's private data to the cloud, 
saves the latency of back-and-forth communication and does not require a network connection.
Many ML applications use on-device ML even for real-time tasks, such as rendering a live video stream~\cite{zhang2020decomposable}, which
is not possible with traditional cloud-based ML on mobile devices.

However, with thousands of private models being deployed on billions of untrusted mobile devices, model theft is a real threat today~\cite{sun2020mind}.  Attackers are not only technically capable of but also financially motivated to steal these models~\cite{sun2020mind, xu2019first}.
Leakage of such proprietary models can cause severe financial loss to businesses -- accurate models help organizations maintain a competitive advantage and training the models requires a significant engineering effort. %(such as data collection, labeling and parameters tuning).

To make matters worse, existing proprietary models are found to be not well protected.
As shown by Sun et al. in \cite{sun2020mind}, 41\% of the models are stored in plaintext and 
can be downloaded along with the application packages.
Applications that protect the models (for example, by encrypting the models~\cite{fritzai}) are themselves vulnerable to
 run-time attacks that can extract the decrypted models from the memory \cite{sun2020mind}. 
Additionally, 54\% ML applications use GPUs for acceleration -- the task of protecting model privacy without losing access to the GPU accelerations is even more challenging.

Prior work  on secure model inference can be classified into two types: 
cryptography based
approaches~\cite{tfencrypted} and trusted execution environment (TEE) based approaches~\cite{darknetz,tramer2018slalom}. 
Both of these techniques face unique challenges for on-device model inference. Prior cryptography based approaches use either 
homomorphic encryption (HE)~\cite{tfencrypted,gilad2016cryptonets} or multi-party computation (MPC)~\cite{juvekar2018gazelle,riazi2018chameleon}. However, HE based techniques are 
orders of magnitude slower than the state-of-the-art (non secure) model inference. 
MPC based approaches involve multiple participants requiring network connectivity which is not suitable for real-time tasks or offline usage. In light of the above challenges, popular mobile platforms, such as Android 7, have made it mandatory to support hardware-backed (TEE) keystore~\cite{android7-cts}. However, prior TEE-based approaches suffer from a lot of drawbacks. First,
the TEE on mobile devices is designed for small critical tasks~\cite{ji2019microtee}, such as key management, while model inference is a
resource demanding task~\cite{juvekar2018gazelle}. Hence, supporting model inference on the limited resources, such as secure memory, of a TEE is challenging. 
Moreover, simply moving the model inference task inside the TEE would significantly increase the TCB size of the TEE (see Sec. \ref{sec:discussion}). An additional problem is the loss of access to hardware accelerators.  
%The state-of-art on-device model protection in industry is pure file-encryption based. Although the models are encrypted at rest, they have to be decrypted in memory before usage. So far, there lacks Hardware-backed (TEE) solution to protecting on-device models. On the other hand, 

\begin{table}[ht]
\centering
\caption{Comparison of ShadowNet with Related Work}
\label{tab:compare_works}
\resizebox{\columnwidth}{!}{
\begin{tabular}{c|c|c|c|c|c}
\hline
    Works & Model Privacy & Mobile TEE & Performance & GPU Access\\
\hline
Slalom \cite{tramer2018slalom} & & & \checkmark & \checkmark  \\
TensorScone \cite{kunkel2019tensorscone} & \checkmark &  & \checkmark& \\
Graviton \cite{volos2018graviton} & \checkmark &  & \checkmark & \checkmark \\
CryptoNets \cite{gilad2016cryptonets} & \checkmark & & & \checkmark \\
TF Encrypted \cite{tfencrypted} & \checkmark & &\checkmark &\checkmark  \\
OMG \cite{bayerl2020offline} & \checkmark &\checkmark & \checkmark & \\
\rowcolor{aliceblue} ShadowNet (Ours) & \checkmark & \checkmark &\checkmark &\checkmark \\

\hline
\end{tabular}
}
\vspace{-2mm}
\end{table}

To this end, we design a novel secure model inference system for convolutional neural networks (CNNs), ShadowNet.
The key idea of ShadowNet is based on the observation that the 
linear layers of CNNs usually take up  $\sim 90\%$ of the computational resources of the whole network. This is in line with previous research, such as Slalom~\cite{tramer2018slalom}. 
ShadowNet offers a novel scheme that allows the heavy linear layers of the model to be 
securely outsourced to the untrusted world (including GPU) for
acceleration without leaking the model weights.
ShadowNet achieves it by transforming the weights of the linear layers before 
outsourcing them to the untrusted world and restoring the results inside the TEE.
The non-linear layers are also kept secure inside the TEE. 
%Table~\ref{tab:compare_works} compares ShadowNet with the other related work.
%which transforms the weights of the linear layers before 
%outsourcing them to the untrusted world for acceleration, and 
%restores the results inside the TEE. The nonlinear layers are also kept secure inside the TEE. 
%With ShadowNet, we can outsource the heavy linear layers to the untrusted world (including GPU) 
%without leaking the model weights.
%Our security analysis shows that, ShadowNet does not give the attacker
%any advantage in learning the model weights.

We build a prototype of ShadowNet based on TensorFlow Lite~\cite{tflite} and OP-TEE OS~\cite{optee}
and evaluate it on five popular CNNs, namely,  MobileNet~\cite{howard2017mobilenets}, ResNet-44, ResNet-404~\cite{he2016deep}, YOLOv4-tiny~\cite{YOLOPaper}%AlexNet~\cite{krizhevsky2012imagenet},
and MiniVGG~\cite{simonyan2014very}.
Our evaluation shows that the ShadowNet's performance overhead is reasonable  --  it increases the model inference time by 0.6$\times$ to 1.6$\times$.
%GPU acceleration helps reduce model inference time from 1ms to 12ms, which has a large room foroptimization with the evolving AI software and hardware accelerators.
%These results are expected as we set the obfuscation ratio at 1.2, which means 
%there is an expected 20\% time increase for the transformed linear layers.
%The extra overhead caused by TEE is also very small, increasing model inference time by 4\% to 11\% 
%for AlexNet and MiniVGG respectively. 
Table \ref{tab:compare_works} compares ShadowNet with related prior work (see Sec. \ref{sec:related_work} for more details). Compared to the cryptography-based approaches~\cite{gilad2016cryptonets} that are usually orders of magnitude slower, ShadowNet provides a practical solution for securing on-device model inference. For instance, for a single image classification, CryptoNets takes around 570 seconds on PC 
while ShadowNet takes $<1$s on a smartphone. 
%We have also explored a variant of ShadowNet scheme, called Layerwise ShadowNet, which applies the ShadowNet transformation
%on a few selected layers to offer a trade-off between security and performance.
%Layerwise ShadowNet benefits significantly from the GPU acceleration and reduces the model 
%inference time by as much as 59ms(26\%) for AlexNet. 

In summary, this paper makes the following contributions:

\squishlist
\item We design a novel on-device model inference system for CNNs, ShadowNet,  which protects 
    the model privacy with a TEE while leveraging the untrusted hardware accelerators. 
\item We build an end-to-end ShadowNet prototype based on 
TensorFlow Lite. We propose novel optimizations to support efficient model inference inside a TEE with a small TCB that can be of independent interest.
\item We build a fully \textit{automated} model conversion tool that can transform a user provided CNN  to the corresponding ShadowNet model. Consequently, ShadowNet models can run seamlessly with user applications on popular ML platforms, such as TensorFlow Lite. %\new{For ARC: architecture-agnostic is a bit exaggerating, we support wide range of architecture, but for YoloV4-tiny, we needs some customization to handle the complex branches/merges. Any idea on how to lower the tone while emphasizing wide architecture support?}
%\item We present a formal security analysis of ShadowNet.
\item Our evaluation on five popular CNNs for mobile platforms, namely, MobileNets, ResNet-44, ResNet-404, YOLOv4-tiny and MiniVGG demonstrates ShadowNet's feasibility for real-world usage on a diverse range of CNN architectures. 
\squishend \vspace{-1mm}
We have open-sourced ShadowNet~\cite{shadownetgithub}.

%% file: background.tex
%\vspace{-1mm}
\section{Background}

\subsection{Convolutional Neural Network}\label{sec:background:CNN} %\vspace{-2mm}
Convolutional neural network (CNN)~\cite{o2015introduction} is a class of deep neural networks which typically consists of an input and an output layer with a sequence of linear and non-linear layers stacked in between. 
The linear layers include convolutional layers and fully connected layers;
the non-linear layers include activation and pooling layers. Some CNNs, such as ResNet, introduce shortcut connections between convolutional layers adding branches and merges in the network structure.
Additionally, MobileNet introduces two new type of linear layers -- pointwise convolution and depthwise convolution. %Object detection networks, such as YOLOv4-tiny, use CNNs as the backbone and have multiple outputs.

\point{Convolutional Layer} %The convolutional layer is the core building block of a CNN. 
%(We will refer to it as the standard convolutional layer.)
The parameters of a convolutional layer consists of
a set of learnable  kernels. Each kernel is characterized by the width, height
and depth of the receptive field. The depth must be equal to the number of channels of the 
input feature map. For the convolutional layer (\textit{conv}) from our example CNN in Fig.~\ref{fig:example},
the input shape is $(222, 222, 3)$ where $3$ is the depth. It has $64$ kernels and the shape of the convolution kernel is $(3, 3, 3)$.%, which corresponds to the receptive field's height, width and depth. 

Let $h$, $b$, $d$ represent the height, width and depth of
the kernel $w$, respectively, and $(x, y)$ refer to the coordinates in the 2D output feature map. Formally, the convolution operation on a given image $I$ with kernel $w$
can be described as follows:
\vspace{-1mm}
\begin{equation} \small
    \mathrm{Conv}(I, w)_{x,y} =  \sum_{i=1}^{h}\sum_{j=1}^{b}\sum_{k=1}^{d}w_{i, j, k}I_{x+i-1, y+j-1, k}
    \label{eq:conv}\vspace{-2mm}
\end{equation}

Let $X$ and $Y$ denote the input and output, respectively, and $W=[w_1,\cdots,w_n]^T$ be the convolution filter. The corresponding convolutional layer is thus given by:
%and the output $Y$ is a three-dimensional matrix.
%\vspace{-2mm}
\begin{equation}\small
    \label{eq:conv}
Y = \mathrm{Conv}(X, W^T)   \vspace{-2mm}
\end{equation}\vspace{-5mm}
\point{Pointwise Convolutional Layer}
For this type of layer, the kernel height and width are both 1.

\point{Depthwise Convolutional Layer}
Depthwise convolution is a type of convolution which applies a single  
convolutional kernel for each input channel.
%In other words, each input channel and filter are binded. 
The number of input channels and the number of kernels are the same. Let $h$ and $b$ represent the height and width of the kernel $w$, respectively, and $(x, y)$ refer to the coordinates in the 
two-dimensional output feature map. 
For a given image $I$ and kernel $w$ which is a 2D matrix, the depthwise convolution $DWConv$ on input channel $c$ can be described as follows:
\vspace{-2mm}
\begin{equation}\small
    \mathrm{DWConv}(I^{(c)}, w)_{x,y} =  \sum_{i=1}^{h}\sum_{j=1}^{b}w_{i, j}I_{x+i-1, y+j-1}^{(c)}
\end{equation}

The depthwise convolutional layer is described as follows,
where $x_i$ represents $i$-th channel of input $X$.
\vspace{-1mm}
\begin{equation}\small
    \begin{split}
        Y &= \mathrm{DWConv}(X, W) \\
          &= [ \mathrm{DWConv}(x_1, w_1), \cdots, \mathrm{DWConv}(x_n, w_n)]
    \end{split}
\end{equation}
%For example, one commenly used activation layer is RELU layer; "average pooling" and "max pooling"
%are two popular types of pooling layer.

%Model weights comprises of the weights of the linear layers and the nonlinear layers.
%The linear layer weights include a group of convolution filters and biases. They occupy the
%majority of the model weights. Most of the nonlinear layers do not have weights, except for a few types
%like the normalization layer.
\point{Dense/Fully Connected Layer}
%The dense or fully connected layer is a common layer used in neural networks as well as CNNs.
%as the last layer for classification. 
The dense layer connects every input node to every
output node. It can be implemented as a pointwise convolutional layer. For example, a dense layer connect
$n$ input to $m$ output can be viewed as a pointwise convolutional layer that has $m$ kernels of size (1, 1, n).

\begin{comment}\subsection{On-device Model Protection}
ML models are core intellectual properties of model owners. Popular ML SDK providers like SenseTime\cite{sensetime}, Face++\cite{facepp}, are highly motivated to protect their models. For example, their business model of licensing ML SDK with offline models requires protection of proprietary on-device models. As shown by recent study on ML model protection in mobile apps\cite{sun2020mind}, model providers invest heavily in on-device model protection solutions, including encrypting model several times, encrypting both model and code, etc. However, their models are still extracted trivially with open-sourced tools. The extracted models are used by tens of popular apps, which are downloaded by millions of users.  

\end{comment}

\vspace{-2mm}\subsection{Trusted Execution Environment}
% TEE and Arm TrustZone: definition, workflow, usage of CA/TA
% limitation: small memory preserved, small TCB to be secure, do not include GPU
% ShadowNet use of 
A trusted execution environment (TEE) is a secure area of the main processor. 
It guarantees the confidentiality and integrity of the code and data loaded inside \cite{tee}.

Arm TrustZone~\cite{armtrustzone} is a popular TEE implementation for mobile devices. 
It is a hardware
feature available on both Cortex-A processors~\cite{armtrustzone-a}
(for mobile and high-end IoT devices) and Cortex-M processors~\cite{armtrustzone-m} (for low-cost
embedded systems). TrustZone creates a  “Secure World”, an isolated environment with tagged caches, banked
registers and private memory, for securely executing a stack of trusted software that includes a tiny OS and trusted applications
(TA). In parallel runs the “Normal World” which contains the regular (untrusted) software stack. Code in the
Normal World, referred to as the client applications (CA), can invoke the TAs in the Secure World. A typical use of TrustZone involves a CA requesting a sensitive service from a TA, such as signing or encrypting data.
Arm TrustZone has been widely used for security critical services, such as key management and 
Digital Rights Management (DRM) on smartphones.

OP-TEE (Open Portable Trusted Execution Environment)~\cite{optee} is an open-source
trusted OS running inside Arm TrustZone. It supports a wide variety
of mobile devices ranging from Arm Juno Board %Raspberry Pi 3 
to a series of Hikey boards.
It is also integrated with AOSP to run alongside Android OS.
OP-TEE OS usually reserves a small part of DRAM (for example, 32MB) as secure memory to minimize the performance impact on the Normal World applications.

%% file: overview.tex
\begin{figure*}[ht]
  \centering
  \includegraphics[width=0.6\linewidth]{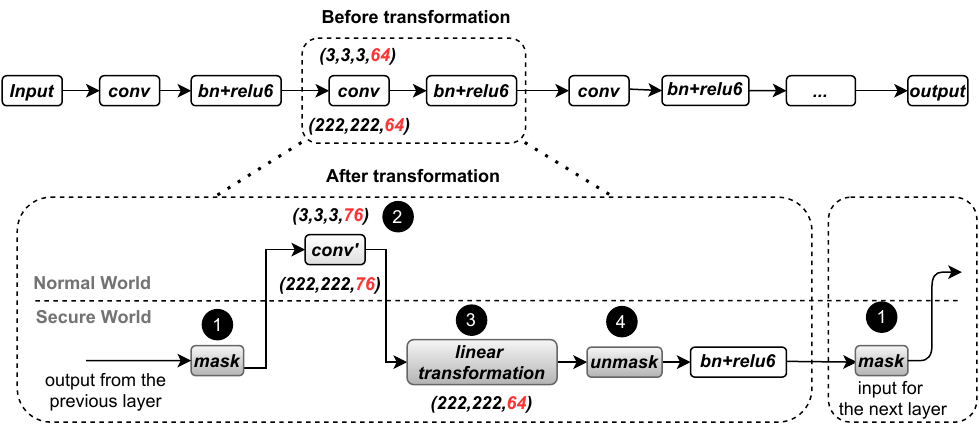}
   \vspace{-1mm} \caption{An overview of ShadowNet transformation on a simple CNN.} \vspace{-2mm}
    \medskip
    \small
    This CNN is a stack of convolutional layers and each convolutional layer (\textit{conv}) is followed by 
    a batch normalization (\textit{bn}) layer and a ReLU6 (\textit{relu6}) activation layer.
    The shapes of the weights are marked on the top of the box and the shapes of the outputs are marked under.
 %   For example, the \textit{conv} layer has 64 kernels of shape (3,3,3), and its output shape is (222, 222, 64). 
    The red color indicates the change in shape after the transformation.
    For each convolutional layer (\textit{conv}), the ShadowNet transformation works in four steps.
%    For each Convolution layer (\textit{conv}), the ShadowNet transformation works in four steps: (1) add a mask layer to the input;(2)
%    replace the original \textit{conv} layer with a transformed \textit{conv'} layer;(3) add a \textit{linear transformation} 
%    layer to restore the result from the \textit{conv'} layer; (4) unmask the input mask. 
%    The combination of \textit{mask}+\textit{conv'}+\textit{linear transformation}+\textit{mask} 
%    is equivalent to the \textit{conv} in the original CNN.
%    The Batchnorm layer and ReLU6 layer remain unchanged. 
    After the transformation, the \textit{conv'} runs in the Normal World, the other layers run in the Secure World. 
   % Note: the input for the first layer is from Normal World, thus not masked, while the output of the first layer is always masked. The output of the last layer will also be output to the Normal World while the input of the last layer is always masked.This is important for the security of the first and the last layer.\par
  \label{fig:example}\vspace{-5mm}
\end{figure*}

\section{Design Overview}
\label{sec:overview}

%\subsection{Problem statement}
%With on-device machine learning becoming more and more popular, thousands of machine learning
%models are deployed on billions of untrusted mobile devices. 
%Mobile apps that use on-device machine learning usually has high-performance demand, like real-time
%image rendering, so it may rely on hardware acceleration, like GPU or NPU.
%Existing solutions include encrypting the model before deployment and decrypt it before usage.
%It suffers from non-sophisticated runtime attack as models can be extracted from memory after decryption.
%So far, there lacks a secure on-device model inference system that can protect the model 
%weights without losing access to the untrusted hardware accelerators.
\vspace{-2mm}
\subsection{Design Goals}
Our secure on-device model inference system has the following design goals:
\squishlist
    \item Security that is rooted in the hardware so that the model remains secure even when the OS is compromised;
    \item Performance efficiency for supporting real-time analysis;
    \item Access to hardware accelerators for supporting on-device ML tasks.
\squishend
\vspace{-3mm}
\subsection{Threat Model}\label{sec:threat}

We consider a strong adversary who controls the Normal World (including the OS) and
observes everything that is exposed to the Normal World (including the GPU tasks). Our primary goal is to protect the \textit{model privacy}. Specifically, the adversary should not learn anything about a model, $\mathcal{M}(\cdot)$, beyond what is revealed by its querying API, i.e. $Y=\mathcal{M}(X)$\footnote{some extra information about the architecture, such as the number and type of linear layers, is also allowed -- see Sec. \ref{sec:sec_ana} for details)}.  A formal security analysis is presented in Sec. \ref{sec:sec_ana}.  We do not consider model inference integrity which can be achieved via verification techniques as proposed in Slalom~\cite{tramer2018slalom}. 
%It means the attacker can launch 
%Denial-Of-Service attack, or corrupt the model inference results in the Normal World, but not be able
%to steal the model parameters.
Additionally, we do not consider side channel attacks on the TEE -- we assume that the TEE can protect the confidentiality 
and integrity of the program and data inside it. Hence, for loading the model into the TEE all standard attestation techniques apply (such as, loading an encrypted model and decrypting it inside the TEE) and this is orthogonal to ShadowNet's design.
%\begin{itemize}
%    \item We assume a strong attacker that takes control of the Normal World including the GPU, 
%        observing everything exposed to Normal World.
%    \item We do not consider model inference integrity but only consider model privacy. The attacker can launch 
%        Denial-Of-Service attack, or corrupt the model inference results in the Normal World.
%    \item We do not consider side channel attack on TEE and trust that the TEE can protect the confidentiality 
%        and integrity of its program and data.
%\end{itemize}
\vspace{-2mm}\subsection{Design Challenges}\label{sec:challenges}
%Naively, we can put the model inference framework in the Secure World or Trusted Execution Environment(TEE).
While TEEs provide hardware-level security, using  mobile TEEs for secure model inference has several technical challenges.
First, mobile TEEs, such as Arm TrustZone, are designed for small security critical
services, such as managing encryption keys. The memory reserved for the TEE OS is limited. For example, only 14 MB is available for trusted applications of OP-TEE OS on Hikey960 Dev Board while the model size of ResNet-404 is 28 MB. 
Hence, it is not feasible to run
the resource-intensive model inference task inside the TEE (see Sec.\ref{sec:discussion}). Second, current TEEs do not 
include the GPU/NPU inside the secure domain. Hence, we would lose access to hardware acceleration. Third, the model inference
framework would also significantly increase the TCB size.

%There are other alternative designs like running a few selected layers inside TEE, which will leave the majority of model
%weights unprotected. The more layers we partition to TEE the less we can benefit from hardware acceleration
%and the more resource we will need from TEE. This simple model partition strategy does not address the three foundmental challenges
%faced by mobile TEE as mentioned above.

\vspace{-2mm}\subsection{Our Solution: ShadowNet}\label{sec:overview:sol}
The key idea of ShadowNet is as follows:
% Use an example to show how ShadowNet transform it into a model running in
% split mode with linear layers outsourced to GPU while nonlinear layers running inside TEE.

\begin{tcolorbox}[sharp corners]\vspace{-2mm}\point{Key Idea}
ShadowNet is based on the observation
that the linear layers of CNNs occupy the majority of the model weight and model inference time~\cite{tramer2018slalom}. For example, we observe that 
the linear layers of MobileNet occupy around 95\% of the model weights and 99\% of the model inference time.
The key idea is to obfuscate the 
 the weights of the linear layers by applying \textit{linear transformations} and outsource them to the 
untrusted world. This enables leveraging the hardware accelerators without trusting them.
ShadowNet then restores the results inside the TEE. The non-linear layers are kept secure inside the TEE.\vspace{-2mm}\end{tcolorbox}
%Only the transformed weights are shared with the
%untrusted world. The parameters used to restore the transformed linear layers and the weights of the
%other nonlinear layers are kept secure inside the TEE. 

%ShadowNet is based on the observation that the linear layers are responsible for most of the computational load and model weights. ShadowNet applies linear transformation on the weights of the linear layers 
%to obfuscate them. Consequently, the heavy linear layers can be outsourced to the untrusted hardware accelerators without leaking the original weights. The nonlinear layers remain unchanged and are computed inside the TEE.
\point{Example Application}
We use a simple example to show how ShadowNet works on a typical CNN as depicted in Fig.~\ref{fig:example}.
The example CNN is a stack of convolutional layers, and each convolutional layer (\textit{conv}) is followed by 
a batch normalization (\textit{bn}) layer and a ReLU6 (\textit{relu6}) activation layer.
%The Convolution layer (\textit{conv}) has 64 kernels of shape (3,3,3), followed by 
%a Batchnorm (\textit{bn}) layer and an activation layer (\textit{relu6}). 
%The output shape of \textit{conv} is (222, 222, 64). 
%The weights shape are marked on the box and the output shape are marked under the box. 

%\begin{figure}
%  \centering
%  \includegraphics[width=0.9\linewidth]{figs/Example-ShadowNet-Lite.pdf}
%    \caption{A lite version of ShadowNet transformation on the example CNN.}
%  \label{fig:example-lite}
%\end{figure}

%We first introduce a lite version of ShadowNet, shown in Figure~\ref{fig:example-lite}. 
For each convolutional layer \textit{conv}, the ShadowNet transformation works in four steps: (1) adds a mask layer to the input; (2)
replaces the original \textit{conv} layer with a transformed \textit{conv'} layer; (3) adds a \textit{linear transformation} layer to restore
the result of \textit{conv'}; (4) unmasks the input. 
The combination of \textit{conv'}+\textit{linear transformation} is equivalent to \textit{conv} in
the original CNN.
The combination of \textit{mask}+\textit{conv'}+\textit{linear transformation}+\textit{unmask} 
is also equivalent to \textit{conv}.
The batch normalization layer and ReLU6 layer remain unchanged. 

The \textit{mask} and \textit{unmask} layers in step (1) and step (4), respectively, are introduced to prevent the adversary from observing the original input and output of the outsourced linear layers. Note that we embed the inputs in a field $\F$ via quantization before applying the mask layer. Additionally, all the weights of the convolutional layers are also quantized accordingly. The unmasked output is de-quantized before
being forwarded to the non-linear layers (for example, the activation layer). We discuss how these layers are
implemented in Sec.~\ref{sec:mask}.
%Step (3) and step (4) show the high-level overview of how ShadowNet works on a convolutional layer.
Note that the \textit{conv'} layer has 76 kernels instead of 64 kernels. This 
is due to the \textit{obfuscation ratio}, a tunable parameter in ShadowNet. 
In Sec.~\ref{sec:scheme_opt}, we explain the rationale behind the choice of this number and
how to generate the weights for \textit{conv'} layer and \textit{linear transform} layer. 
%We will also discuss how ShadowNet transforms the other type of linear layers, namely, pointwise convolutional, depthwise convolutional and dense/fully connected layers. 
%\lcm{centering this long caption does not look well}
\begin{figure}[!ht]\vspace{-2mm}
  \centering
  \includegraphics[width=0.6\linewidth]{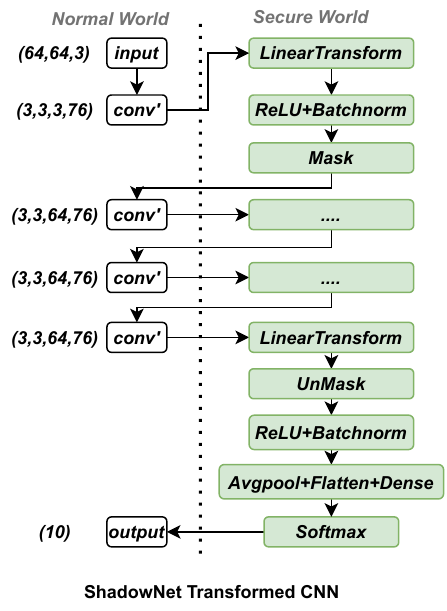}\vspace{-2mm}
    \caption{\textbf{An example of ShadowNet transformed CNN.}}
     The above example has four convolutional layers. %The shapes of the input, output and weights of the convolutional layers are marked in the figure. 
     Observe that for the first layer, \textit{only the output} is masked. For the last layer, \textit{only the input} is masked. This is because according to our threat model, the input to the first layer (model input) and the output of the last layer (model output) is known to the adversary.
  \label{fig:example-cnn}\vspace{-6mm}
\end{figure}

\point{Discussion} In summary, ShadowNet offers a novel model inference system that protects 
the model weights with a TEE while leveraging the untrusted hardware for acceleration. 
ShadowNet achieves its goal by transforming the computationally heavy linear layers' weights and masking their input 
before outsourcing them to the untrusted world and restoring the results inside the TEE.  All the non-linear layers are kept secure inside the TEE. With this design, \squishlist
    \item ShadowNet's security is rooted in the TEE, meeting the first design goal. 
    \item ShadowNet does not introduce any heavy cryptographic operations and our evaluation shows that ShadowNet is efficient
 -- this meets our second design goal. 
 \item 
ShadowNet is still able to use the accelerators which 
%and will be continuously benefited from the fast evolving
%hardware accelerators, 
meets the third design goal. 
\squishend
ShadowNet solves the technical challenges (Sec. \ref{sec:challenges}) of mobile TEEs by maintaining low memory usage and a small TCB which is detailed in Sec. \ref{sec:scheme} and Sec. \ref{sec:implementation}. 

%% file: design.tex
\vspace{-1mm}\section{ShadowNet} \vspace{-1mm}
\label{sec:scheme}
In this section, we introduce ShadowNet.
First, we explain how ShadowNet applies linear transformation on a broad class of
linear layers, namely convolutional, pointwise convolutional, depthwise convolutional, 
and dense/fully connected layers. Next, we discuss the mask layer for input/output privacy. Finally, we describe an optimized implementation of the linear transformations for ShadowNet.  \\
\textbf{Notations.} Here, we introduce the notations we use for the rest of the paper. $\mathbb{F}$ denotes a field. $X$ and $Y$ denote the input and output of a convolutional layer. $X'$ and $Y'$ denote the masked input and output. $\hat{W}=\big[\hat{w_1},\cdots,\hat{w_n}\big]$ represents the transformed convolution filter corresponding to the original convolution filter $W=\big[w_1,\cdots,w_n\big]$ where $\hat{w}_i \thinspace (w_i)$ denotes the masked (original) convolution kernel. %Additionally, all random values (used for masking and the transformations) in ShadowNet are drawn uniformly at random from an interval $[-t,t]$ for some $t$\footnote{$t$ is sufficiently large and kept  secret.} $\in \mathbb{R}$.
For a positive integer $n \in  \mathbb{N}$, $[n]$ denotes the set $\{1,\cdots,n\}$. 
\subsection{Quantization}\label{sec:quantization}
All the transformations in ShadowNet work on a field $\F$. For this, we quantize all inputs and weights of a CNN to integers and embed them in the finite field $\mathbb{Z}_p$ modulo a prime $p$ ($p$ is sufficiently large to avoid wrap-around). This step is necessary for providing a formal security guarantee (see Sec. \ref{sec:sec_ana}). Following prior work \cite{tramer2018slalom,Gupta2015}, ShadowNet converts a floating point number $x$ to a fixed-point representation as $\tilde{x} = \textsf{FP}(x; l) := \textsf{round}(2^l \cdot x)$. After the computation of the linear layers, we de-quantize $\tilde{x}$ by scaling it by $2^{-l}$.
\subsection{Transformation of Linear Layers}
\label{sec:transform}
ShadowNet relies on linear transformation to obfuscate the weights of the linear layers.

\point{Linear Transformation}
Linear transformation is a function $f$ defined on vector spaces $V$ and $T$ over
the same field $\mathbb{F}$, $f : V \rightarrow T$. %In our scenario, $R$ is the field of real number.
For any two vectors $u, v \in V$ and any scalar $c \in \F$,
the following two conditions are satisfied: 
\vspace{-1mm}
\begin{equation}\small
    \begin{aligned}
        additivity:   f(u + v) &= f(u) + f(v) \\
        homogeneity: f(cu) &= cf(u)
    \end{aligned}\label{eq:linear_transformation}\vspace{-2mm}
\end{equation}

%We apply linear transformation on the weights of linear layers, then
%the linear operation is performed on the input and the transformed weights. 
%We apply linear transformation on the result to transform it back
%to what it should be when the operation is performed on the original weights. 
%The above property guarantees that the result on the input and the transformed
%weights is equal to the transformed result on the input and the original weights.

% what's linear transformation in mathmatic form

\point{Convolutional Layer} A convolutional layer is given by $Y = \mathrm{Conv}(X, W^T)$ where $X$ and $Y$ denote the input and output, respectively, and $W=[w_1,\cdots,w_n]^T$ is the convolution filter. A detailed discussion of convolutional layer is attached in Appendix \ref{sec:background:CNN}. Let $F=[f_1,\cdots,f_n]$ be a random filter such that each $f_i, i \in [n]$ has the same shape as $w_i$. Additionally, let $\Lambda$ be a diagonal matrix where the diagonal elements $\lambda_i \in \F, i \in [n]$ are random scalars. Conceptually, the linear transformation on the convolutional layer in ShadowNets  works as follows: 
\vspace{-2mm}\begin{equation}\small
    \label{eq:conv_w_to_t}
    \begin{split}
     \hat{W}^T &= W^T \cdot \Lambda  + F\\\Rightarrow
        [\hat{w}_1,\cdots, \hat{w}_n] &=  [w_1, \cdots,w_n]
                              \begin{bmatrix}
                        \lambda_1 & & \\
                         & \ddots & \\
                         & & \lambda_n \\
        \end{bmatrix} + [f_1,\cdots,f_n] 
    \end{split}
\end{equation}

Thus, each of the kernels are transformed as follows \begin{gather}\hat{w}_i=\lambda_i w_i+f_i\end{gather}

\begin{comment}
Now, note that

\begin{equation}\small
    \begin{split}
        [w_1, \cdots, w_n] &=  [\hat{w}_1-f_1, \cdots,\hat{w}_n-f_n]
                              \begin{bmatrix}
            \frac{1}{\lambda_1} & & \\
                         & \ddots & \\
                         & & \frac{1}{\lambda_n} \\
        \end{bmatrix}  \;\text{or} \\
                              W^T &= (\hat{W}^T - F) \cdot \Lambda^{-1}
    \end{split}
\end{equation}
\end{comment}

Hence, from the above transformation (Eq.~\eqref{eq:conv_w_to_t}) and the properties of linear transformation (Eq. \eqref{eq:linear_transformation}),  we have:\vspace{-1mm}
\begin{equation}\small
    \begin{split}
        \mathrm{Conv}(X, \hat{W}^T) &= \mathrm{Conv}(X, W^T.\Lambda + F)) \\
        & = \mathrm{Conv}(X, W^T) \cdot \Lambda + \mathrm{Conv}(X,F)
    \end{split}\vspace{-2mm}
\end{equation}

So, we can compute $\mathrm{Conv}(X, \hat{W}^T)$ on the untrusted GPU and restore the output $Y$ inside the TEE as follows:\vspace{-1mm}
\begin{equation}\small
    \label{eq:recover}
    Y = \big(\mathrm{Conv}(X, \hat{W}^T) - \mathrm{Conv}(X,F)\big)\cdot \Lambda^{-1}\vspace{-2mm}
\end{equation}%\vspace{-2mm}
Note that the computation of $\mathrm{Conv}(X,F)$ is done in the Normal World. We discuss an optimized implementation of  the above in Sec. \ref{sec:scheme_opt}.
\point{Pointwise Convolutional Layer}
The scheme for the standard convolutional layer can be directly
applied to the pointwise convolutional layer.

\point{Depthwise Convolutional Layer}
For depthwise convolutional layers, ShadowNet applies linear transformations on both the input and the kernels. 
Specifically, we $(1)$ shuffle the sequence of input/kernel channels and  $(2)$ obfuscate each input/kernel channel with a random scalar as detailed below.

Assume that the input has $n$ channels. Thus, the depthwise convolutional layer has $n$ kernels, one per channel.
Let $w_i$ represent the $i$-th kernel of the convolution filter $W$, where $W = [w_1, w_2, \ldots, w_n]^T$. 
Let $(\lambda_1, \ldots, \lambda_n) \in \F^n$ be a set of random scalars and $\pi \in S_n$ ($S_n$ is the group of permutations on $[1, \ldots, n]$) be a random permutation.
Additionally, let $P_{\pi}$ be a $n \times n$ permutation matrix corresponding to $\pi$ ($P_{\pi}(i, j) = 1$ if $\pi(i) = j$, and $0$ otherwise).
We scale and shuffle the sequence of the kernels in $W$ with $\Lambda$ as follows:
\vspace{-1mm}
\begin{equation}\small
    \begin{split}
        \Lambda &= \begin{bmatrix}
                        \lambda_1 & & \\
                         & \ddots & \\
                         & & \lambda_n \\
        \end{bmatrix} 
    \end{split} \; P_{\pi}\vspace{-2mm}
\end{equation}
We apply the same permutation to shuffle the input channels.
The input transformation matrix $A$ is defined as follows: 
\vspace{-1mm}\begin{equation}\small
    \begin{split}
        A &= \begin{bmatrix}
            \lambda_1^{-1} & & \\
                         & \ddots & \\
                         & & \lambda_n^{-1} \\
        \end{bmatrix} 
    \end{split} \; P_{\pi}\vspace{-2mm}
\end{equation}
Thus, for identity matrix $I$, we have:
\vspace{-2mm}
\begin{equation}\small
    A \cdot \Lambda ^ T = I \vspace{-2mm}
\end{equation}
Given the transformed weights $\hat{W} = [\hat{w}_1, \ldots, \hat{w}_n]^T$, the transformations on the input $X$ and weights $W$ are described  
as follows:
\vspace{-3mm}
\begin{equation}\small
    \begin{split}
        \hat{W}^T &= W^T \cdot \Lambda\\
        X' &= X \cdot A
    \end{split} \vspace{-3mm}
\end{equation}

Let  $Y' = \mathrm{DWConv}(X',\hat{W}^T)$. It is easy to see that: \vspace{-1mm}
\begin{equation}\small
    Y' = [ \mathrm{DWConv}(x_1, w_1), \cdots, \mathrm{DWConv}(x_n, w_n)] P_{\pi} \vspace{-1mm}
\end{equation}

Let  $P_{\pi}^{-1}$ be the inverse of $P_{\pi}$. 
We can restore the correct result with the following equation:
\vspace{-2mm}
\begin{equation}\small
    \label{eq:dw_restore}
    Y = Y' \cdot P_{\pi}^{-1}\vspace{-3mm}
\end{equation}

%\begin{equation}
%        B = 
%                              \begin{bmatrix}
%                                  & & 1 \\
%                                  & \iddots &\\ 
%                                  1 & &  
%                              \end{bmatrix}
%\end{equation}
%\begin{equation}
%        B = 
%                              \begin{bmatrix}
%                                  0 &1 &\ldots &0 \\
%                                  0 & 0 & \ldots & 1 \\
%                                  & &\vdots & \\ 
%                                  1 &0 &\ldots &0
%                              \end{bmatrix}
%\end{equation}

%As we can see, the scales in $A$ and $\Lambda$ are carefully chosen so the actual scaling will be 1,namely no scaling. 
%The sequence shuffling in $\Lambda$ and $B$ will be combined, so the input will be inverted twice. 

%We can easily verify that the transformed depthwise convolutional layer is equivalent to the
%original depthwise convolutional layer.
%
%\begin{equation}
%    \begin{split}
%        Y' \cdot B & = DWConv(X',T) \cdot B \\
%          &= [ DWConv(x_2, w_2), \cdots, DWConv(x_1, w_1)] \cdot B \\ 
%          &= [ DWConv(x_1, w_1), \cdots, DWConv(x_n, w_n)] \\ 
%         & = Y \\
%    \end{split}
%\end{equation}
%
%It is worth mention that the generated shuffled sequence is just an example for the ease of demonstration.
%In real system, we can shuffle the input channels randomly and keep the sequence as a secret. 
Note that both the transformation of the input and the restoration of the result are
performed inside the TEE while the depthwise convolution on the transformed filter
can be outsourced to the untrusted GPU. 

%It is worth mention that the special linear transformation on different channels can also be expressed as a
%pointwise convolutional layer. We can treat the transformation matrix $\Lambda$ as $n$ convolutional kernel of shape(1,1,n).
%So conceptually, we are inserting extra pointwise convolutional layers into the original network.

\point{Dense/Fully Connected Layer}
%The dense or fully connected layer is a common layer used in neural networks as well as CNNs.
%as the last layer for classification. 
Recall that a dense layer can be implemented as a pointwise convolutional layer. Hence, ShadowNet applies the same linear transformation as the one described for the standard convolutional layer.
\vspace{-2mm}
\subsection{Layer Input/Output Privacy}
\label{sec:mask}\vspace{-0.8mm}
We introduce the mask/unmask layer to protect the input $X$ and output $Y$ of the convolutional layers.
%Else, an adversary can infer the weights of the filter $W$ by solving a system of linear equations after observing  sufficient number of $(X,Y)$ pairs. This is discussed in more detail in Section~\ref{sec:sec_ana}.

\point{Mask Layer}
The mask layer adds a random mask to the input of a convolutional layer which is outsourced to the Normal World. Note that, in a typical CNN, the output of a convolutional layer output will be the input for the next layer. In the offline phase, ShadowNet generates random masks $M$ of the same shape as $X$ inside the TEE. The masked input $X'$ is defined as follows:\vspace{-3mm}
\begin{equation}\small 
    X' = X + M \vspace{-2.5mm}
\end{equation}
$X'$ is outsourced to the Normal World for the convolutaional layer with filter $W$. A fresh mask is used for every convolutional layer and for every round of model inference.
Since all the values are embedded in a field $\F$, this masking is equivalent to applying a one-time pad~\cite{KatzLindell}. 
\point{Unmask Layer} After obtaining the masked output $Y'$ from the Normal World,  the TEE restores the original value via: \vspace{-2mm}
\begin{equation}\small Y=Y'-\mathrm{Conv}(M,W)\vspace{-2mm}\end{equation}%\vspace{-2mm}
where the TEE pre-computes and stores the value of $\mathrm{Conv}(M,W)$ in an offline phase. This masking step is similar to Slalom \cite{tramer2018slalom}. 

%This masking step is essentially an stream cipher encryption scheme 
Recall from our threat model (Sec. \ref{sec:threat}) that the adversary already knows $Y=\mathcal{M}(X)$ where $\mathcal{M}$ is the CNN model. Hence, the input to the first layer (model input) and the output of the last layer (model output) is \textit{not} masked. %Only, the output (input) of the first (last) layer is masked. 
This is depicted in Fig. \ref{fig:example-cnn}. After $Y$ is unmasked inside the TEE, it is de-quantized and then, forwarded to the non-linear layers.\vspace{-0.5mm}
\subsection{Optimized Implementation}
\label{sec:scheme_opt}

In this section, we describe how ShadowNet implements the linear transformation for the convolutional layers. 
\par Recall from our discussion in Sec. \ref{sec:transform} that $F$ acts as a mask that protects the weights of the kernels.  However, the TEE needs access to $\mathrm{Conv}(X,F)$ for restoring $Y$ (Eq. \eqref{eq:recover}) -- the extra computation for $\mathrm{Conv}(X,F)$ is an performance overhead. This introduces a performance/security trade-off which is tackled in ShadowNet as follows:
\squishlist \item First, we select $r \in \mathbb{R}, r > 1$.  We refer to $r$ as the \textit{obfuscation ratio} and it is a parameter for tuning the performance/security trade-off. We elaborate on this later in this section.   \item Select at random $F=[f_1,\cdots,f_{m-n}]$ where $m=\lceil r\cdot n\rceil$. Each $f_i$ has the same shape as that of the kernels of $W$. \item Select a set of $n$ random scalars $(\lambda_1,\cdots,\lambda_n)\in \F^n$. \item Compute  \vspace{-1mm}\begin{equation}\small
    \label{eq:expand_weight}
    W' = [\lambda_1 w_1 + f'_1, \dots , \lambda_{n} w_n+f'_n, f_1, \dots, f_{m-n}]^T\vspace{-2mm}
\end{equation} 
where $f'_i$s are randomly chosen from $F$. Repetitive choice is allowed here as $m-n$ 
%\footnote{In case of $r=2$ or $m=2n$, we do not need repetitions and can have $f'_i=f_i$.} 
might be smaller than $n$. 
\item Store an index matrix $C$ 
%\footnote{We don't need the index matrix for $r=2$ since $C[i]=i$.}  
where $C[i]=j$ iff $f'_i=f_j$.
\item Finally, shuffle $W'$ with a random permutation matrix $P_{\pi}$ for $\pi \in S_m$.
\vspace{-2mm}\begin{equation}\small
    \label{eq:shuffle_weight}
    \hat{W}^T = W'^T P_{\pi} \vspace{-2mm}
\end{equation}\vspace{-2mm}
\squishend
All of the above steps can be pre-computed securely in an offline phase. With this transformation, we can easily recover the convolution results with the inverse of the permutation matrix $P_{\pi}$, the index matrix $C$ and
the random scalars $(\lambda_1, \ldots, \lambda_n)$. Conceptually, the recovery process can be implemented as a pointwise convolution with $n$ filters of shape $(1,1,m)$.
\par
Intuitively, the permutation $\pi$ prevents the adversary from distinguishing between the kernels in $\hat{W}$ that correspond to $F$ and the ones that correspond to (transformed) $W$. Clearly, higher the values of $r$,  better is the security and higher is the computational overhead. %Specifically, for $r=2$ (equivalently, $m=2n$), the weights of the true kernels are perfectly hidden at the cost of $n$ extra convolution operations of the form Equation \ref{eq:conv}. 
The formal security and performance analysis is presented in Secs. \ref{sec:sec_ana} and \ref{sec:obf-ratio}, respectively. Note that ShadowNet applies the aforementioned transformation to every convolutional layer of the CNN. %fresh random values ($M$, $\pi$, $F$, $\lambda_i$, $\Lambda$) for every layer and every round of model inference.

\noindent\textbf{Discussion.} Here, we discuss how ShadowNet's design addresses the challenges of mobile TEE as outlined in Sec. \ref{sec:challenges}. First, ShadowNet tackles the challenge of limited TEE memory by running only a subset of the layers of the model inside the TEE. Specifically, it leaves the resource-heavy linear layers outside the TEE. To further reduce ShadowNet's computational load, 
we propose the aforementioned optimized design which reduces TEE's overhead for restoring the value of the transformed linear layers. Additionally, we propose several novel optimizations for ShadowNet's implementation to reduce its memory consumption (Sec. \ref{sec:ca-ta}). Second, our design of the transformations provides a formal security guarantee (Thm.~\ref{thm:1}). Consequently, ShadowNet allows efficient outsourcing of the linear layers to the accelerators without compromising on security. Third, the new operations introduced by ShadowNet (for masking and linear transformations) are relatively simple. Hence, ShadowNet's implementation adds only $2,100$ LOC into the TCB (see Sec. \ref{sec:implementation}). Note that TensorFlow Lite has tens of thousands of lines of code and porting it as a whole inside the TEE would require a much larger TCB. ShadowNet, by design, keeps it outside the TCB, thereby supporting its rich functionality in a secure manner with a small TCB.

%% file: impl.tex
\vspace{-0.2cm}\section{Implementation}\label{sec:implementation}

%Proposed section struture:
%\begin{itemize}
%    \item Overview of the prototype including the workflow of adopting ShadowNet
%    \item Model convertion
%    \item TensorFlow extention
%    \item OP-TEE CA/TA for ShadowNet
%    \item Performance optimization(reduce memory usage; enchance speed;)
%\end{itemize}

%Our prototype uses Android P as the Normal World OS and OP-TEE OS as the Secure World/TEE OS.
%We choose Hikey960 board as our development board as it is similar to 
%a state-of-art smartphone and offers unlocked TEE development environment for Arm TrustZone. 
%It has four ARM Cortex-A73
%and four Cortex-A53 cores with 4GB SRAM and a Mali G71 MP8 GPU.
% Ruimin: {the above is repeated in the experiment setup in the evaluation. Consider remove one of them.}
% 

\begin{figure}
  \centering
  \includegraphics[width=0.6\linewidth]{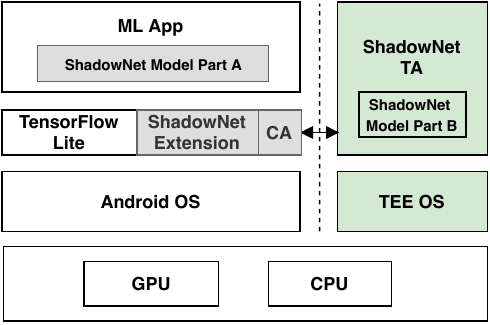}
    \caption{\textbf{The system architecture of Shadownet on mobile platforms.}
    %\medskip
    %  \small 
    Color grey shows the part modified by ShadowNet in the Normal World.
       Color green shows the part in the Secure World.  
       The ShadowNet model part A and B refer to
       linear and non-linear part, respectively. }
  \label{fig:system}\vspace{-5mm}
\end{figure}

%#In this section,
%#we start with an overview of the ShadowNet prototype, then 
%#we discuss three key components of the ShadowNet prototype: ShadowNet model conversion, TensorFlow Lite extension, 
%#and ShadowNet CA/TA.

\subsection{Overview of the ShadowNet prototype}
We implement the ShadowNet prototype as an end-to-end on-device model 
inference system as shown in Fig.~\ref{fig:system}.  Our prototype contains the ML mobile application, TensorFlow Lite runtime library 
with extensions for ShadowNet, ShadowNet client application (CA) and ShadowNet trusted application (TA).  ShadowNet runs the transformed linear layers in the Normal World and
uses the CA to communicate with the TA in the Secure World, 
which runs the other (non-linear) layers securely. 
ShadowNet models are split into two components -- Part A and Part B which run in the Normal World  
and the Secure World, respectively. During inference, Part A, which runs in TensorFlow Lite, processes the linear layers in the Normal World. 
When executing the non-linear layers, Part A invokes the CA to send commands to the TA which runs Part B inside the Secure World and passes the results back.
\vspace{-2mm}
\subsection{Model Conversion}\label{sec:implementation:conversion}
The model conversion mechanism can be divided into two steps as depicted in Fig.~\ref{fig:model-converter}.
%\arc{inconsistency in description}
\squishlist\item \point{Step I} The original model (left of Fig.~\ref{fig:model-converter}) is converted to the ShadowNet model (middle of Fig.~\ref{fig:model-converter}) with the help of \textit{four} new layers that transform the weights of the linear layers. %This model is functionally identical to the final model under deployment except it contains both the linear and the non-linear layers. 
\item \point{Step II} All the non-linear layers from the above ShadowNet model (red dotted in Fig.~\ref{fig:model-converter}) are replaced with a placeholder layer which represents the model in the Normal World and interacts with the Secure World. Now, we have two components of the model -- Part A contains the weights of the transformed linear layers and runs in the Normal World (right of Fig.~\ref{fig:model-converter}) while Part B contains the weights of the other layers in the Secure World (red dotted box in Fig.~\ref{fig:model-converter}). 
\squishend
%Given the original TensorFlow model, which can be a Python script with Keras API,  
%we follow a set of rules to transform the description into a full ShadowNet model description. 
%Based on it, we 
%(1) derive the model description with only linear layers and generate ShadowNet Model Part A. 
%(2) extract the weights of the nonlinear layers and generate ShadowNet Model Part B.

%The workflow of model conversion 
%is depicted in Figure~\ref{fig:shadownet-workflow}.
%Given a description for the original TensorFlow model, which can be a Python script with Keras API,  
%we follow a set of rules to
%transform the description into a full ShadowNet model description. 
%On the left branch,
%we derive its model description with only linear layers and generate ShadowNet Model Part A. 
%On the right branch,
%we extract the weights of the nonlinear layers and generate ShadowNet Model Part B.

%\begin{figure}[ht!]
%  \centering
%  \includegraphics[width=0.7\linewidth]{figs/workflow-model-conversion.pdf}
%    \caption{The workflow of model conversion}
%  \label{fig:shadownet-workflow}
%\end{figure}

\point{Step I} Recall that, convolutional, pointwise convolutional, depthwise convolutional and fully connected/dense layers are considered to be the linear layers. We transform these layers in Step I by introducing four new  layers, namely, \textit{LinearTransform}, \textit{ShuffleChannel}, \textit{PushMask} and \textit{PopMask}.
\textit{LinearTransform} applies a linear transformation to the input.
\textit{ShuffleChannel} scales each channel in the input and shuffles their sequence.
\textit{PushMask}/\textit{PopMask} adds/removes a mask from its input respectively. %\new{Fig.~\ref{fig:model-converter}} gives a concrete example.
%\zs{Agree with changming, consider remove the details in the squishlist below.}
%\lcm{These bullets points are too detailed and repetitive given we already outlined the idea with the above text.}
In a nutshell,
\squishlist
    \item Every convolutional layer (including pointwise convolutional and fully connected/dense layers) is replaced  with four layers: \textit{PushMask} layer, transformed convolutional layer, \textit{LinearTransform} layer and \textit{PopMask} layer.
    \item Every depthwise convolutional layer is replaced with five layers: \textit{PushMask} layer, \textit{ShuffleChannel} layer, transformed depthwise convolutional layer, \textit{ShuffleChannel} layer and \textit{PopMask} layer.
    \item The weights for the convolutional layers and mask layers are quantized. The outputs from quantized layers are de-quantized accordingly before being forwarded to the non-linear layers.
    \item The batchnorm layers can be fused with its preceding convolutional layers for the ease of implementation (we do not show it in the example for the ease of exposition).
\squishend

%We use the name \textit{PushMask} and \textit{PopMask} to show that the mask and unmask layers are used in pairs.

\begin{figure}[h!] \vspace{-2mm}
  \centering
  \includegraphics[width=0.7\linewidth]{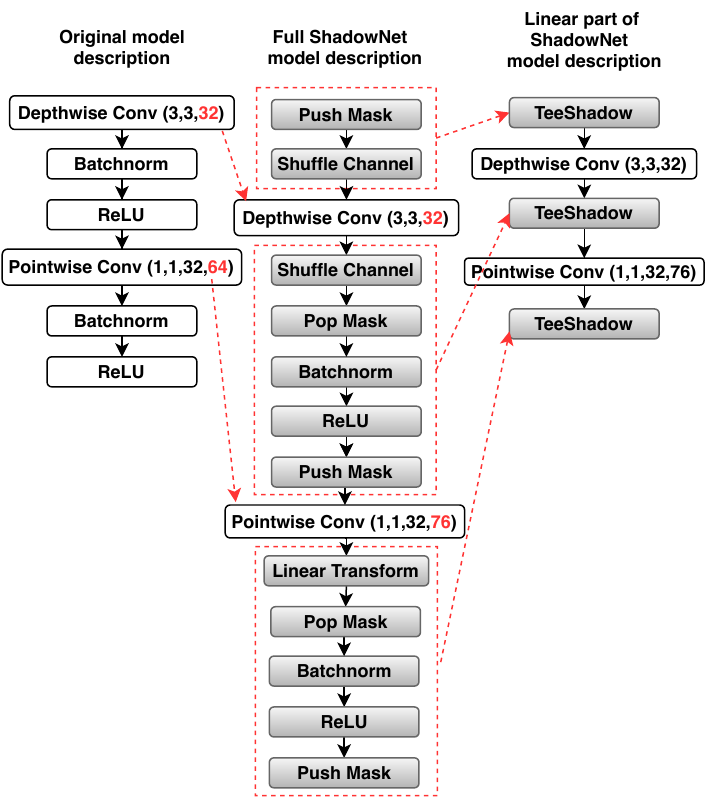}
    \caption{\textbf{An example of ShadowNet model conversion ($r=1.2$).} First, the
    convolutional layers are transformed. The pointwise convolutional layer's weight shape changes from (1, 1, 32, 64) to
    (1, 1, 32, 76), where $76 = 64 \times 1.2$. Next, the linear layers remain unchanged and the
    non-linear layers are replaced with a placeholder \emph{TeeShadow}. }
  \label{fig:model-converter}\vspace{-3mm}
\end{figure}

%\point{Model conversion rules}
%Model conversion has two steps. First, we 
%transform a given model description into a full ShadowNet model description; 
%Second, we derive the linear part of the ShadowNet model from the output of the first step. 
 
%All the other types of layers are considered nonlinear. 

%a) Nonlinear layers remain unchanged;
%b) For any Convolution layer (Pointwise Convolution and Dense layer included), replace it with four layers: PushMask layer, % obfuscated Convolution layer, LinearTransform layer, and  PopMask layer; %c) For any Depthwise Convolution layer, replace it with 5 layers: PushMask layer; ShuffleChannel layer; %obfuscated Depthwise Convolution layer, ShuffleChannel layer and  PopMask layer.

%During conversion, the weights for the ShadowNet model should be generated properly so that the converted
%ShadowNet model will be equal to the original model. 

%ShadowNet offers tools to help convert the model weights.
%For example, the weights conversion tool takes the original weights of the 
%Convolution layer and generate the weights for a transformed Convolution layer and a transformation matrix for the 
%LinearTransform layer with the algorithm described in Appendix~\ref{sec:alg}. 

%When we derive the linear part of the ShadowNet model,  
\point{Step II} In Step II, we introduce an additional layer called \textit{TeeShadow} that replaces all the non-linear layers between two successive linear layers. This newly introduced layer serves as a placeholder for the interposed non-linear layers. During model inference, \textit{TeeShadow} switches to the Secure World to handle the non-linear layers and returns back to the Normal World after its completion.
Note that all the newly added ShadowNet layers in Step I %(LinearTransform, ShuffleChannel, PushMask and PopMask) 
as well as the original non-linear layers are placed in the Secure World.  

 CNNs with shortcuts, such as ResNet~\cite{he2016deep} and InceptionNet~\cite{szegedy2015going},  contains branches and merges of different convolutional operations in its data-flow.  We address this by introducing a new operation called \emph{TeeMerge} as shown in Fig~\ref{fig:resnet}. 
Unlike \emph{TeeShadow} which takes a single input, \emph{TeeMerge} can take multiple inputs and produce the corresponding output.

 \begin{figure}
   \centering
   \includegraphics[width=0.9\linewidth]{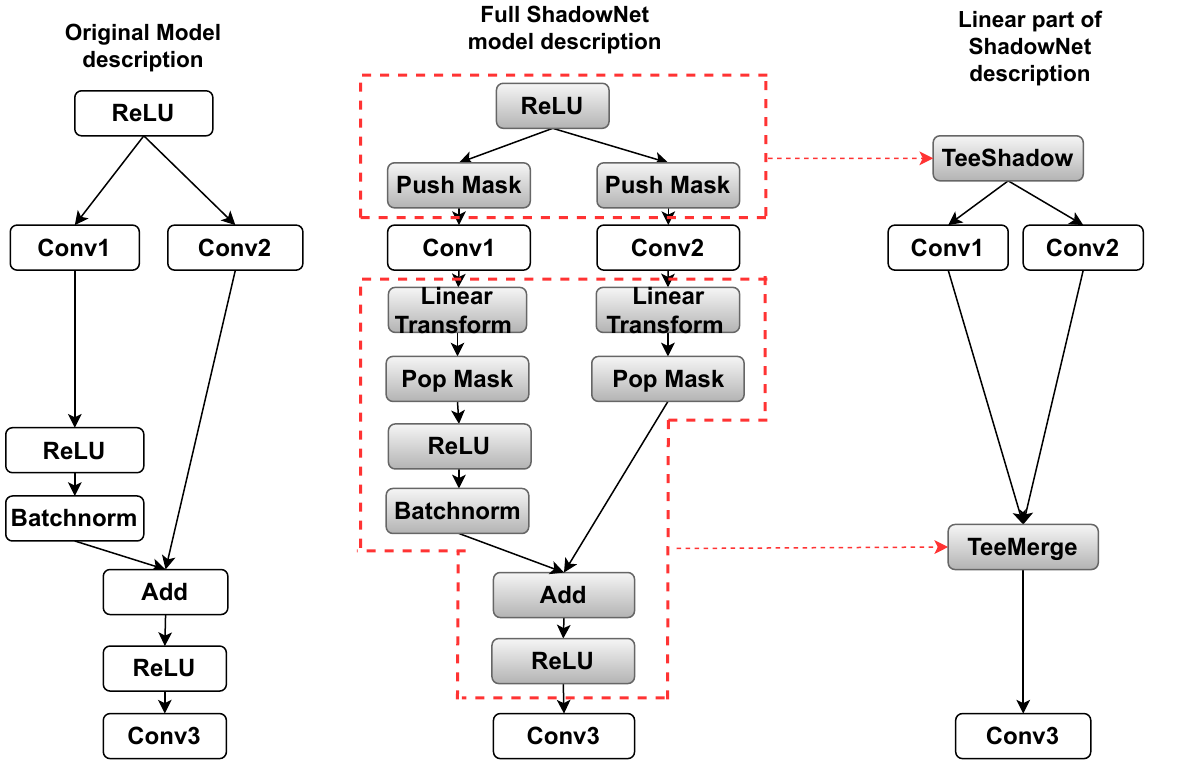}
     \caption{\textbf{An example of TeeMerge.} In the left part of the figure, a shortcut is introduced by ResNet (a convolution is removed from the left branch for simplicity).
    The middle figure shows the ShadowNet transformed model where the outputs from two different convolutions (\emph{Conv1} and \emph{Conv2}) 
    need to be merged. A new  layer, \emph{TeeMerge}, is introduced to take the inputs from the two branches and generate the corresponding output. 
     }
   \label{fig:resnet}\vspace{-5mm}
 \end{figure}
In a nutshell during Step II, the non-linear layers are replaced with \textit{TeeShadow} layers and the linear layers remain unchanged, as shown in Fig.~\ref{fig:model-converter}.

\vspace{-2mm}
\subsection{Adding ShadowNet Support in TensorFlow Lite} \label{sec:TF}

%TensorFlow is a Machine Learning framework where developers can define, train and run
%a model. TensorFlow Lite is based on
%TensorFlow, and only contains the model inference framework 
%for mobile and embedded devices. 
%It does not support defining or training a new model. 
%To use TensorFlow Lite, you need to convert an existing model to TensorFlow Lite format.
%ShadowNet needs TensorFlow  support to define the ShadowNet model with new layer types.
%ShadowNet also needs TensorFlow Lite support to run on mobile devices.

ShadowNet introduces five new operations that are not inherently supported in TensorFlow Lite (TF Lite), namely, \textbf{\textit{LinearTransform}}, \textbf{\textit{ShuffleChannel}}, 
\textbf{\textit{AddMask}}, \textbf{\textit{TeeShadow}} and  \textbf{\textit{TeeMerge}}. Note both PushMask and PopMask layers are implemented via \textbf{\textit{AddMask}}. 
To implement these operations, we extend TF Lite as follows:
 \squishlist
     \item We implement these operations as CustomOps for TensorFlow.
     \item We add these CustomOps in the Keras Layer API.
     \item We add the TF Lite implementations of the CustomOps.
%     \item Modify the TensorFlow Lite model converter, which converts TensorFlow model to TensorFlow Lite model, to support the
%         convert of ShadowNet models.
 \squishend

%Adding ShadowNet support in TensorFlow (Lite) is non-trivial. There are several challenges. First,
%\textbf{\textit{LinearTransform}} and \textbf{\textit{ShuffleChannel}} 
%involve computation on sparse matrix used by us, for which no existing compute library can be reused in TensorFlow (Lite). 
%Second, \textbf{\textit{TeeShadow}} operation has different
%output shapes even on the same input shape, so the output shape inference required by TensorFlow (Lite) 
%for resource allocation needs to be handled correctly. Third, the implementation of TeeShadow for TensorFlow Lite 
%involves interaction with the Trusted Application(TA) in the Secure World, requires 
%integrating TEE client library with the TensorFlow Lite library. 
%Fourth, converting a ShadowNet model from TensorFlow format to TensorFlow Lite format is not supported
%with the latest TensorFlow Lite model converter. We need to add support for the ShadowNet models.
% The tool is not fully stabilized. Different version
% of the tool would generate different output for the same input model,
% and the latest version generate a wrong model which has assertion error.
% We find a workaround by switching to a older version.

The extension is based on TensorFlow 2.2 which was the latest version at the time of our implementation.
To support CustomOp for TensorFlow, we add $563$ LOC in Python, $924$ LOC in C++; to support TF Lite, we add
$774$ LOC in C++. Our tool for model conversion is $1223$ LOC in Python.

%With the above support, we can use TensorFlow Keras API to write our ShadowNet model descriptions, and use ShadowNet
%tools to convert existing models into ShadowNet models. More importantly, for mobile apps using TensorFlow Lite models,
%they can switch to ShadowNet models seamlessly with ShadowNet-enabled TensorFlow Lite library.
%It is transparent to the user applications. Ideally, the developers do not need to change the application code at all.
\begin{tcolorbox}[sharp corners] \vspace{-2mm}\point{Remark 1} The motivation of integrating ShadowNet into TF Lite is two fold. First, it shows that it is feasible to apply ShadowNet on a mainstream on-device model inference platform. Second, adoption of ShadowNet simply requires updating the TF Lite library with our ShadowNet patch. This introduces no changes to the application logic, i.e., the model loading/inference interface remains the same as that of the standard TF Lite. Thus, our design of ShadowNet makes it practical for real-world usages. \vspace{-2mm}\end{tcolorbox}
\vspace{-4mm}
\subsection{ShadowNet CA and TA}
\label{sec:ca-ta}
%\point{The role of ShadowNet CA and TA}
The ShadowNet CA and TA work as follows. 
%ShadowNet CA handles the communication between the TensorFlow Lite runtime library
%and the ShadowNet TA.
%We implement a pair of ShadowNet CA and TA for OP-TEE OS to handle the secure model inference.
%ShadowNet CA is the client application in the Normal World,
%connecting TensorFlow Lite's TeeShadow operation and the TA. 
%When it is first called by the TeeShadow operation, it prepares
During the initialization phase, the CA starts a secure session with the TA
and loads the ShadowNet model Part B into the TA. Before each round of model inference, the
CA loads the pre-computed weights for the mask layers into the TA. 
During  model inference, the CA 
passes the parameters from the \textbf{\textit{TeeShadow}} operation to the TA and fetches results from it.
The TA performs the model inference task for the non-linear layers represented by the \textbf{\textit{TeeShadow}} (see Fig.~\ref{fig:model-converter}). 

The CA sets up a secure session with the TA and passes the parameters. The ShadowNet TA is implemented as a lightweight and generic model interpreter that can parse the model binary (Part B), and run partial model inference based on the CA's parameters. It can handle any model in the TensorFlow Lite format (flatbuffers).

We add $480$ and $2,100$ LOC in C for the CA and the TA, respectively. Specifically, the new operations  (\textbf{\textit{LinearTransform}}, \textbf{\textit{AddMask}}, \textbf{\textit{ShuffleChannel}}) add $<200$ LOC in the TCB as they are relatively simple operations.

%ShadowNet TA is the trusted application running inside TEE for the nonlinear layers. 
%As shows in Figure~\ref{fig:model-converter},
%every linear layer is followed by a group of nonlinear layers. Each group is treated as a subnet
%with unique input and output shape, depending on the position of the linear layer. For each subnet,
%there is a corresponding TeeShadow operation in the Normal World.
%% Ruimin: {what's the difference between TeeShadow and the `subnet` here?}
%Before model inference,
%the initialization includes: a) Initialize all of the subnets' structure for the nonlinear layers;
%b) Allocate memory buffer for the weights and load them into buffer.
%During model inference, it takes the output from a certain linear layer as its input and %run the inference of the following %nonlinear layers and send the results back to the CA.

%Different networks have different structures. 
%Ideally, we should auto-generate the CA and TA based on the description of
%the network structure. So far, we have not automated it yet in our prototype. 
%We write a customized CA/TA pair for each different network. 
%Fortunately,
%the CA/TA pairs for different networks share many code in common. 

\subsection{Optimizations}
\label{sec:implementation:opt}

\point{CA and TA's communication}
The OP-TEE OS is designed to work with a TA with a relatively small TEE memory usage ($<1$MB).
As a result, for every call the CA makes to the TA, the TEE OS sets up a page mapping for the entire memory used by the TA even when the CA and the TA are in the same session. This has an adverse effect on the performance -- when the TA's memory size increases, the associated cost of the CA's call to the TA increases proportionately even when the TA is idle inside the TEE. For example, the time increase from 0.1ms to 6ms for a memory increase of 1MB to 64MB. This TEE design limitation is known to be a major challenge in the general OP-TEE community~\cite{optee-issue1,optee-issue2}. In order to tackle this, we propose a novel optimization as follows. Note that only the parameters (handles for the memory shared between the CA and the TA) of the CA's call change and require a new mapping each time -- the majority of the TA's memory, namely the data and code, do not need to be re-mapped for every call.
%-- the TA's code and data remains the same. 
Based on this observation, we implemented an optimization in the TEE OS to cache the page mapping for the TA's code and data section. Only the pages corresponding to the parameters passed between the CA and the TA are re-mapped.  This improves the performance significantly. For instance for ResNet-44, the ShadowNet model inference time reduces from $106$ms to $57$ms -- a $2\times$ improvement. We believe that our key idea of caching certain pages could be of independent interest to the broader OP-TEE community.
\point{Optimizing TA's Memory Management}
The TEE OS has a limited memory reserved for the TA. 
Without careful memory management, the TA would exhaust the memory and crash. We tackle this challenge as follows.
%TA has very limited memory, 
First, we allocate memory statically in the TA to avoid fragmentation caused by dynamic memory allocation. Specifically, for a given model, the memory needed for each layer can be determined at the time of loading and allocated statically. Second, we do not allocate memory for the output of each layer. Instead, we allocate two sufficiently large buffers and rotate them to save memory.
\point{Optimizing TA's Performance}
Implementing the TA in OP-TEE has many restrictions. For example, 
since OP-TEE only supports C,
%while existing open-source Machine Learning frameworks are either in
%C++ or Python. 
we would lose access to  popular compute libraries, such as Eigen\cite{eigen}
and Arm Compute Lib \cite{armcomputelib} in C++. 
Additionally, OP-TEE lacks a math library. Hence, 
we propose efficient implementations of mathematical functions, such as \textit{sqrt}, \textit{exp}, 
or \textit{tanh}, for the non-linear layers. 

% Memory allocation: dynamic allocate vs static allocation
% Implementation: batchnorm layer, 
% float point computation:
% hardware 

%\point{Fusing nonlinear layers}
%Fusing continuous layers is a common optimization technique. By fusing operations, we can 
%avoid extra memory copy and improve performance. 
%Continus nonlinear layers that sharing similar data-access pattern can be fused into one operation.
%Element-wise operation like activation layer can always be fused with previous
%layer. Our ShadowNet TEE backend will always fuse continuous layers whenever possible.

\begin{table}[ht]\vspace{-1mm}
\centering
    \caption{Optimizations of the TA (with MobileNets model)}
\label{tab:optimization}
\begin{tabular}{l|c}
\hline
    Optimizations & Exec. Time (ms)  \\
\hline
    Baseline (Static mem. alloc)& 1500  \\
    %Newlib sqrt& 300 & Compile TA with \textit{-O2} \\
    (1) Neon sqrt &300 \\
    (2) Cache friendly& 245  \\
    (3) Optimize loop sequence & 205 \\
    (4) Preload weights & 100 \\
    (5) Neon for \textit{Batchnorm}, \textit{AddMask}& 90 \\
    (6) Neon for ReLU6 & 81 \\
\hline
\end{tabular}
\footnotesize
{\\
\raggedright \textit{Note}: The optimizations are applied incrementally in a sequence. For example, $81$ms is
    the execution time inside the TA (excluding mask weights reloading time) when optimizations (1) to (6) are all applied.}\vspace{-1mm}
\end{table}

%\point{Other TA optimizations}
Initially, we ported the non-linear layers from  Darknet~\cite{darknet13}, a deep learning framework for desktop written in C. However, the resulting TA was very slow on our Arm64 Dev Board. Hence, we propose the following optimizations to bring the performance at par with that of TensorFlow Lite. \squishlist
\item Use Arm Neon to optimize the \textit{sqrt} implementation in the Batchnorm layer. 
\item Swap the inner and outer loops when necessary to support cache-friendly data
access for multi-dimensional data. %\textcolor{blue}{For example, when accessing multi-dimension arrays in nested loops, the inner loop should always handles the last dimension where data are stored continuously in memory.} 
%\textcolor{red}{ARC: Address RA question -"Which inner and outer loops are being swapped?"} 
\item Move repetitive computation out of the loops and pre-compute it. 
\item Pre-load all the weights to avoid repetitive loading.
\item Use Neon multiply+add instructions to optimize the \textbf{\textit{Batchnorm}} and \textbf{\textit{AddMask}} operations. 
\item Use
Neon minimum/maximum instructions to optimize the activation layers, such as  ReLU6.
\squishend Table~\ref{tab:optimization} shows the execution times for our proposed optimization for the MobileNets TA. An illustration of the \textit{sqrt} optimization is presented in App.~\ref{sec:sqrt}
%Among them, preloading weights, cache awareness, loop sequence adjustment and \textit{sqrt} optimization have 
%contributed a lot during the optimization.
%We attach the details for }, which is an example of how we profile and optimize the TAs.

\point{Configuration of the TEE OS}
We also increase the size of the secure memory reserved for the TEE OS from 16MB to 64MB to accommodate a larger Part B. Additionally, we changed the size of the reserved shared memory from 2MB to 8MB. This shared memory is used for communications between the Normal World and the Secure World. These changes only require a few lines of configurational code in the TEE OS and Arm Trusted Firmware. No change is required in the Normal World OS, such as Android/Linux.

% Ruimin: {add some conclusion of the results.}

% Ruimin: {Other notes: I wonder if Table tab:optimization and tab:sqrt should stay here or be moved to the evaluation section. I think the results are interesting and should be discussed more thoroughly. For example, how much time reduction is achieved after applying each layer or type of optimization? A specific optimization is most useful because half of the layers can use it. If the target model is different from MobileNets, for example, model X that has more Y layers, then optimization Z will be more useful than the others. }

%% file: evaluation.tex
\vspace{-1mm}\section{Evaluation}
\label{sec:eval}
%\arc{add implmentation details of YOLO, add alll new results}
%\zs{added in the introduction of "models"}
%The system reserves 16 MB RAM for TEE OS, of which 14 MB can be used by TA.
%Our ShadowNet system is based on TensorFlow 2.2.

Our evaluation focuses on four questions:
\squishlist
    \item Correctness: Does ShadowNet produce the same result as the original model?
    \item Efficiency: What is the overhead introduced by ShadowNet?
    \item Obfuscation Ratio: What is the impact of the obfuscation ratio on the correctness and performance of ShadowNet?
    \item Security: How secure is ShadowNet?
\squishend
%We will answer the above questions in this section.
\vspace{-2mm}

\point{Configuration} %We use Hikey960as our development board which has Arm TrustZone support. 
We perform the evaluation on the Hikey960~\cite{hikey960} board equipped with Kirin 960 SoC
with 4 Cortex A73 + 4 Cortex A53 Big.Little CPU architecture, ARM Mali G71 MP8,
and 3GB LPDDR4 SDRAM.
We run Android P in the Normal World and OP-TEE OS 3.9.0~\cite{optee-aosp} in the Secure World. We use the field $\mathbb{Z}_p$ for prime $p=2^{24}-3$ and a fixed-point representation of $l=8$ for quantization. The obfuscation ratio is set to be 1.2. % We discuss the impact of obfuscation ratio in Section~\ref{sec:obf-ratio}.}
\point{Models} We evaluate ShadowNet on five popular models -- MobileNet~\cite{Mobilenets_Code}, ResNet-44~\cite{Resnet44_Code}, ResNet-404~\cite{Resnet404_Code}, YOLOv4-tiny~\cite{YOLOv4-tiny_Code} and MiniVGG~\cite{Minivgg_Code}. The rationale behind choosing these models is that they cover a wide range of CNN architectures. MiniVGG is derived from VGG which represents a classical CNN architecture.
 MobileNet uses pointwise  and depthwise convolution which are convolutional layers customized for mobile devices. ResNet uses shortcut connections between the
convolutional layers which create branches and merges in
the network structure.  YOLOv4-tiny is an object-detection model with a complex CNN architecture. First, object detection requires a multi-task model/multi-objective optimization which outputs a prediction class and a bounding box. Second, the model structure is non-sequential with new activation operations, such as LeakyRelu~\cite{DBLP:journals/corr/RedmonDGF15}, and complicated Lambda layers consisting of different non-linear operations, such as  upsampling, concatenation. Third, the original model size is 34MB which is large in the context of mobile environments. 

\point{Datasets} MobileNet is evaluated on the ImageNet-2012 dataset \cite{ILSVRC15} with 50K images. ResNet-44, ResNet-404 and MiniVGG are evaluated on the CIFAR-10 dataset \cite{cifar10} with 10K images. 
We evaluate YOLOv4-tiny on the VOC2007 testing dataset~\cite{voc}.

\begin{figure}
  \centering
  \includegraphics[width=0.75\linewidth]{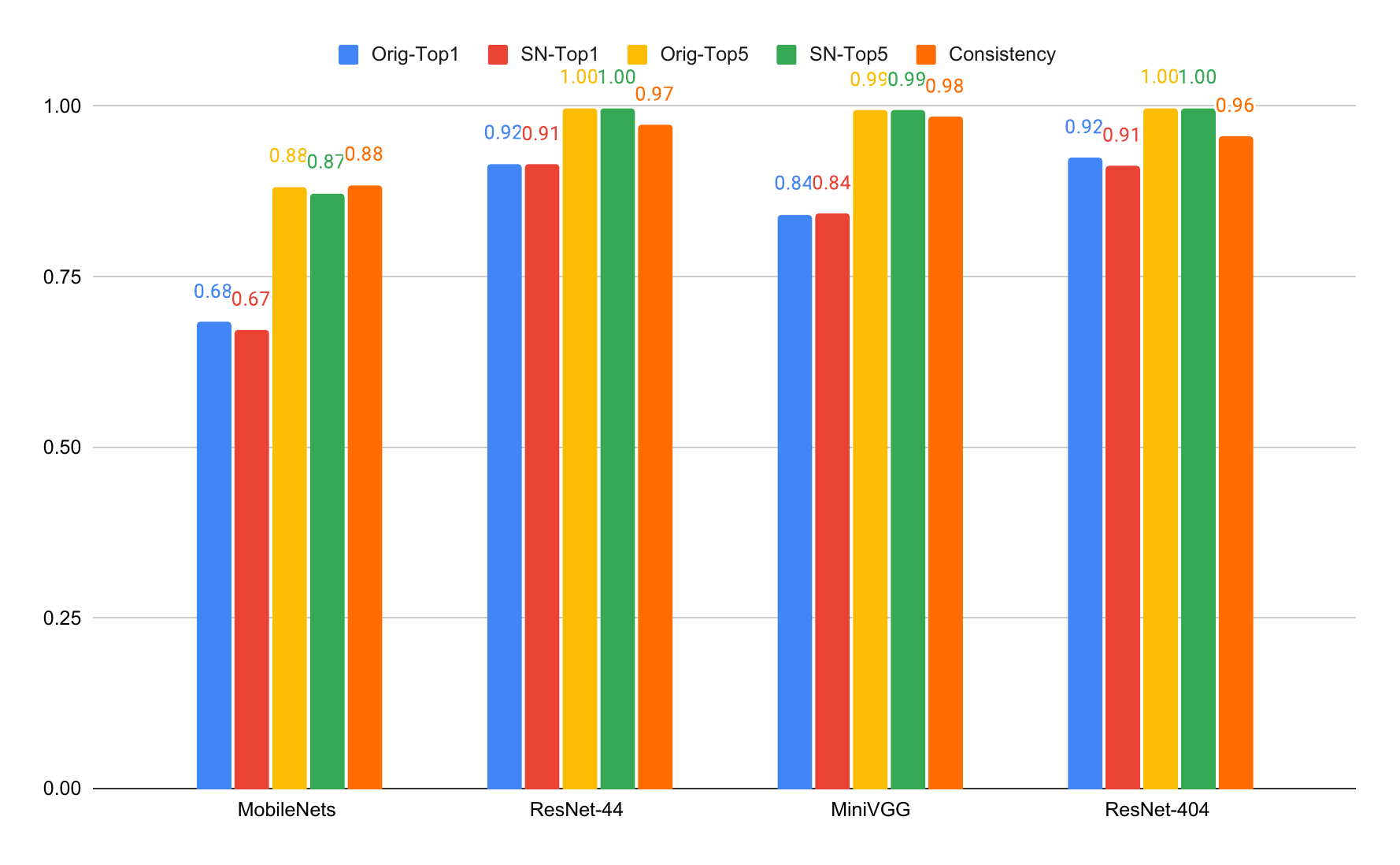}\vspace{-2mm}
    \caption{\textbf{Prediction accuracy and consistency evaluation.}
    Orig-Top1 refers to the original model's top-1 accuracy, SN-Top1(5) refers
    to the top-1(5) accuracy of the ShadowNet transformed model. Consistency measures the \% of ShadowNet predictions that match with that of the original model.
    }
  \label{fig:accuracy}\vspace{-8mm}
\end{figure}

\vspace{-2mm}\subsection{Correctness}\label{sec:correctness} %\vspace{-2mm}
We evaluate the correctness of ShadowNet by comparing the prediction accuracy and consistency before and after model transformation. 
Consistency checks whether ShadowNet top-1 prediction is consistent with the original model on the same input.
Fig.~\ref{fig:accuracy} shows our evaluation results on four models -- MobileNet, ResNet-44, ResNet-404 and MiniVGG.
 For YOLOv4-tiny, we calculated mean average precision (mAP) based on 50\% Intersection Of Union which is the standard metric for object detection. The accuracy of the original model and ShadowNet is $57.06\%$ mAP and $55.90\%$ mAP, respectively. Our implementation is based on the pre-trained model from \cite{YOLOv4-tiny_Code}. Overall, we observe that ShadowNet has accuracy comparable to that of the original models ($\sim1\%$ decrease is due to the numerical errors from quantization -- an essential step for security). 
 
 ShadowNet is $94\%$ consistent for YOLOv4-tiny.  ShadowNet's Top1 prediction consistency for ResNet-44, ResNet-404, MiniVGG and MobileNet are $97\%$, $96\%$, $98\%$ and $88\%$, respectively. The relatively low consistency for MobileNet is due to the inputs which are predicted \textit{incorrectly} by the original model. Specifically for inputs with correct predictions from the original model, ShadowNet is $95\%$ consistent; for the incorrect predictions, ShadowNet is $74\%$ consistent. The original model's mean confidence (top-1 score) for the inputs with correct and incorrect prediction is $0.8$ and $0.38$, respectively. Recall that MobileNet is evaluated on ImageNet with 1K classes. Hence for classification tasks with such a large number of classes, ShadowNet's numeric errors can impact inputs that have low confidence scores from the original model (for instance, the top-2 scores are very close). The performance is acceptable since this affects mostly the inputs with \textit{incorrect} predictions from the original model. 

\subsection{Efficiency}\label{sec:efficiency}
%We evaluate the efficiency of ShadowNet on MobileNet, ResNet-44, \new{ResNet-404, YOLOv4-tiny} and MiniVGG. 
We use the model inference time as our metric which measures the time span between feeding an image
and getting the classification result.
%Then, we measure the performance impact of the TEE communication and GPU acceleration. 
%jWe also evaluate an interesting variant of 
%ShadowNet scheme, the Layerwise ShadowNet, for which we apply ShadowNet scheme on a few selected linear layers.

\point{Experimental Highlights}
Our evaluation shows that:\squishlist \item ShadowNet results in a reasonable overhead, 
increasing the model inference time by $0.6\times$ to $1.6\times$. 
%A significant part of
%the overhead comes from refreshing the mask layers' weights. 
%ShadowNet is much more practical comparing with the crypto-based weights protection schemes, whose overhead are usually orders of magnitude higher.
%mode. The user experience of real-time object detection is not affected in a noticeable way. 
\item  GPU acceleration can reduce the model inference time for all three models from $1$ms to $30$ ms. 
There is still a large room for improvement with software and hardware updates.
%as the ShadowNet models forces the execution to be split between CPU and GPU.
%This can be improved by optimizing the CPU/GPU communications in ShadowNet. 
\vspace{-4mm}
\squishend
\point{Methodology}
We use the TensorFlow Lite Android image classifier Demo application \cite{tfdemo} developed by Google to evaluate the end-to-end
model inference time. We evaluate ShadowNet under different settings --
(1) the original model on CPU; (2) the ShadowNet transformed model on CPU; (3)
the original model on GPU; (4) the ShadowNet transformed model on GPU.
%\end{itemize}
%\begin{itemize}
%\item the original model on CPU \item the ShadowNet transformed model on CPU \item
%the original model on GPU \item the ShadowNet transformed model on GPU
%\end{itemize}
% TODO obfuscation ratio impact
\\
Fig.~\ref{fig:sn-perf} shows the model inference time for ShowdowNet.  

%Comparing with the original model, ShadowNet model adds extra nonlinear layers, 
%increases the computation for linear layers, and also adds communication overhead between
%CA and TA. The extra computation is inherent to ShadowNet scheme while the TEE communication
%overhead varies on how the CA/TA is implemented. 

\begin{figure}
  \centering
  \includegraphics[width=0.75\linewidth]{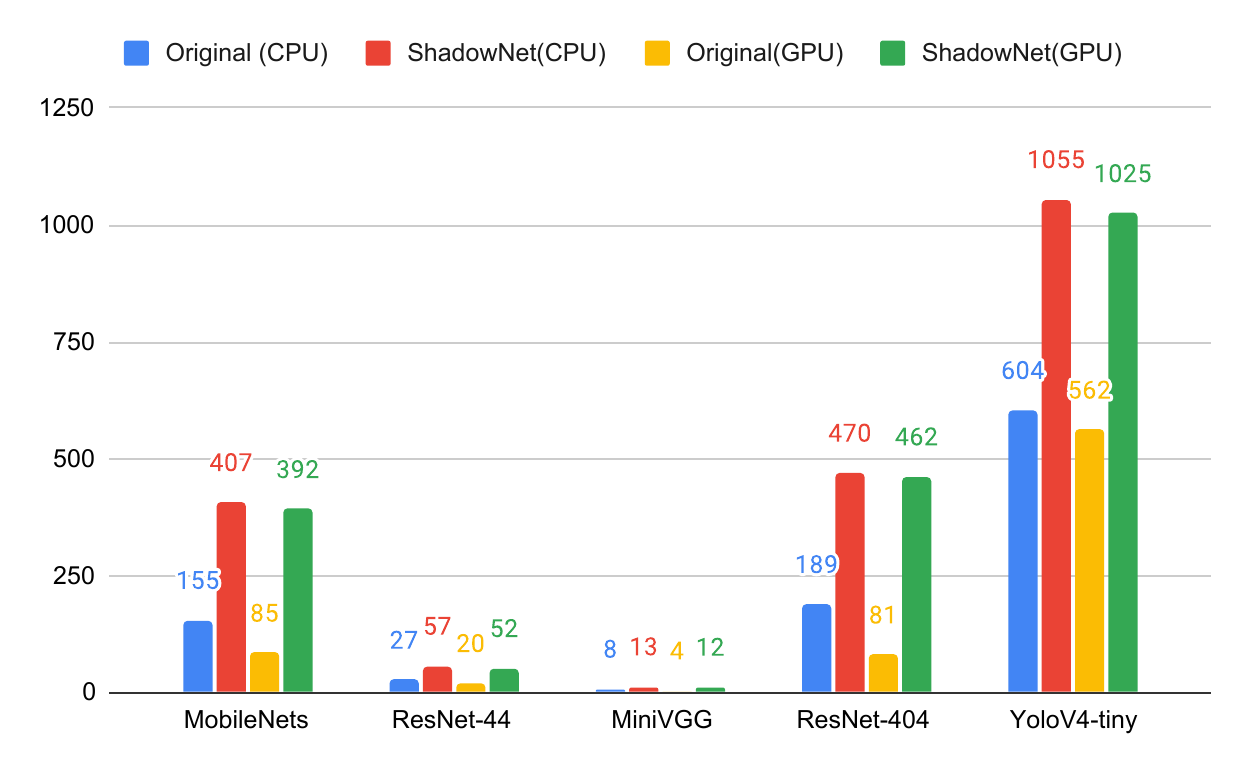}
 \vspace{-3mm}   \caption{\textbf{ShadowNet performance evaluation (ms).} 
    We evaluate ShadowNet on five CNN models: MobileNet, ResNet-44, MiniVGG, ResNet-404 and YOLOv4-tiny ($r=1.2$).
    For each model, we measure the model inference time for one image on four different settings:
    original model on CPU, ShadowNet model on CPU, original model on GPU and ShadowNet model on GPU. We report the mean of 100 trials. 
    }
  \label{fig:sn-perf}\vspace{-5mm}
\end{figure}

\point{ShadowNet's Performance on CPU}
Compared with the original model in the CPU mode,  ShadowNet 
incurs an overhead of 252ms $(1.6\times)$, 30ms $(1.1\times)$, 5ms $(0.6\times)$, 281ms ($1.5\times)$ and 451ms $(0.7\times)$ 
for MobileNet, ResNet-44, MiniVGG, ResNet-404 and YOLOv4-tiny, respectively.
The overhead is reasonable as ShadowNet refreshes the masks for each round of model
inference. Table~\ref{tab:sd-size} shows the model size before and after the ShadowNet transformation. 
For example, for MobileNet, the weights of the mask layers has a size of $37$MB
which takes around $200$ms to be updated. For ResNet-$44$, the mask size is $3.3$MB which takes $18$ms to be updated. 

%\zs{Do we need to emphysize our optimization on the CA/TA switch here? We should have explained it in the implementation section.}
%The Normal World and Secure World (TEE) switches can also cause extra overhead; there are 26 and 44 world switches for MobileNets and ResNet-44, respectively.
%This is because the current design of the TEE OS redoes the mapping of the TA's pages for every world switch. A possible optimization would be to cache the page mappings inside TEE, which can reduce the overhead by hundreds of milliseconds. We leave this as part of future work.

%As a result, its performance is less affected by ShadowNet (61\%).
%ShadowNet increases the model inference time by 293\% for ResNet-44 which requires significantly more world switches (44 in total). 
%Compared to the crypto-based  approaches which are orders of magnitude slower\cite{gilad2016cryptonets}, the performance of ShadowNet is quite good. 

\begin{table}[!ht]
\centering
\caption{ShadowNet model size change.}
\label{tab:sd-size}
\resizebox{\columnwidth}{!}{
\begin{tabular}{c|c|c|c|c|c|c}
\hline
    \multicolumn{2}{c|}{Model Size(MB)} & MobileNet & ResNet-44 & ResNet-404 & YOLOv4-tiny & MiniVGG \\
    \hline
    \multirow{2}{*}{\begin{tabular}[c]{@{}c@{}}Original\\ Model\end{tabular}} &   \textit{Conv} & 17 & 2.43 & 24.71 & 22.56 & 20 
    \\\cline{2-7} & \textit{Other} & 0.175 &0.029 & 0.23 & 0.79 &  0  \\ 
    \hline
    \multirow{2}{*}{\begin{tabular}[c]{@{}c@{}}ShadowNet\\ Model \end{tabular}} & \textit{Conv} & 20& 2.92 &   29.79   & 27.11 &24  
    \\\cline{2-7} & \textit{Other} & 37 &3.3 &29.73 & 40.53 & 2  \\
\hline
\end{tabular} \vspace{-2mm}
}
\footnotesize
    {
        \\Model weights are divided into two parts -- \textit{Conv} layers' weights (standard convolutional, pointwise convolutional,
    depthwise convolutional and dense layers) and \textit{Other} layers' weights (\textit{AddMask}, \textit{LinearTransform}, \textit{ShuffleChannel}, \textit{Batchnorm}).
    Among the \textit{Other} layers, the weights of the \textit{AddMask} layer occupy the maximum space for ShadowNet models.} \vspace{-2mm}
\end{table}

Despite having an additional 5ms to 451ms latency depending on the original model size, we find that the impact on the user is acceptable. We test this by running Demo to detect objects in real-time. For MobileNet,  Demo processes around six images per second for
the original  model, and two and a half images for ShadowNet.
For small models, such as MiniVGG, the extra $5$ms latency is imperceptible to users.

%Comparing with ShadowNet without TEE in CPU mode,  ShadowNet with TEE 
%only incurs 3ms (11\%), 9ms (4\%) and 11ms (7\%) extra overhead for MiniVGG, AlexNet and MobileNets, respectively.
%%The extra overhead of TEE mode is caused by the 
%%TEE communication and the customized CA/TA implementation for nonlinear layers.
%This shows that the extra overhead introduced by TEE is negligible comparing with the ShadowNet transformation.
%Based on this observation, we measure the Layerwise ShadowNet scheme without involving TEE,
%avoiding the engineering effort for implementing the customized CA/TA pairs.

\point{Performance Impact from GPU}
Running on GPU reduces model inference time by 15ms (4\%), 5ms (9\%) and 1ms (8\%), 8ms (2\%), 30ms (3\%) for
MobileNet, ResNet-44, MiniVGG, ResNet-404 and YOLOV4-tiny, respectively.  The benefit of GPU acceleration
for ShadowNet model is not as good as that for the original models due to the following reasons. First, the GPU speedup ($\sim 2\times$) on our evaluation board is significantly less than the GPU speedup 
in the cloud. Any advancement in on-device GPUs would immediately improve ShadowNet's performance. Second, ShadowNet requires splitting the model inference between the CPU (TEE mode) and the GPU. The interleaving between the CPU and the GPU causes extra overhead for repeated setting up of the GPU jobs.  
%This switching overhead, in fact, nullifies the performance gain from the GPU for ResNet-44.

We would like to remark that mobile GPU acceleration for on-device ML is still a developing research area~\cite{mobilegpu}. 
Hence, ShadowNet's benefit is that it opens the door for leveraging future advances in on-device accelerators.

\point{Offline Time} Offline model conversion takes $10$min $59$s for MiniVGG, $6$min 5s for ResNet-44, $31$min $23$s for MobileNet, $28$min $46$s for ResNet-404, $24$min $23$s for YOLOv4-tiny. Regeneration of mask layers' weights for $100$ instances takes $2.19$s for MiniVGG, $20.29$s for ResNet-44, $97.26$s for MobileNet. $127.9$s for ResNet-404, $25.5$s for YOLOv4-tiny.
\begin{comment}
%\point{Comparison with Other TEE-based Solutions} 
%Here, we compare the performance of ShadowNet with two existing TEE-based secure ML systems: Graviton~\cite{volos2018graviton} andTensorScone~\cite{kunkel2019tensorscone}. Graviton~\cite{volos2018graviton} implements a TEE
inside the GPU,  allowing secure ML tasks to run in an isolated part of GPU with an overhead of 17-33\%. 
TensorScone~\cite{kunkel2019tensorscone} is a secure ML system 
built on the Intel SGX-based secure container called SCONE\cite{arnautov2016scone}. The reported model 
inference time of a single image with InceptionV4\cite{szegedy2015going} model is around $1.18$ s
under TensorScone SGX hardware mode. Without TensorScone, the native model inference time on CPU is around $0.35$ s.
Graviton is more efficient than ShadowNet as it changes the GPU and provides security at hardware level. 
TensorScone uses Intel SGX, and it has worse performance (around 235\% overhead) than ShadowNet
as it runs the whole model inference inside SGX which also suffers from limited secure memory. 
TensorScone does not support GPU acceleration. ShadowNet overcomes this limitations with a novel scheme that requires less 
memory from TEE while outsourcing the linear layers securely to GPU.
\end{comment}
\vspace{-2mm}
\subsection{Obfuscation Ratio}\vspace{-1mm}
\label{sec:obf-ratio}
%In Section~\ref{sec:transform}, we introduced the obfuscation ratio $r$ which has a range of $[1.0, 2.0]$. 
%Our scheme remains secure under different $r$.  
%Various $r$ can also help hide the original shape of the weights. For example, the attacker will not know the actual number filters for the Convolution layer by observing the transformed layer. 
Fig.~\ref{fig:obf_ratio} shows
how the accuracy, size and inference time of the MobileNet model vary with $r$. 
We observe that both \textit{top1} and \textit{top5} accuracy remain almost unchanged as $r$ varies from $1.0$ to $2.0$.
This shows that ShadowNet's accuracy does not depend on the obfuscation ratio. 
Another observation is that the model size increases linearly with $r$ and so does the model inference time. 
This is intuitive because both the weight size and the amount of computation for the convolutional layers grow linearly 
with $r$. See Sec. \ref{sec:sec_ana} for more discussion on how to choose $r$. 

\vspace{-2mm}\subsection{Security Analysis}\vspace{-1mm}
\label{sec:sec_ana}
\input{attack}
%\arc{Why we need it this evaluation? Why we chose these two attacks?}
In this section, we present the formal\footnotetext{We are ignoring the batch normalization and ReLU layers for simplicity and a worst case analysis for security. In practice, the adversary can only observe $\hat{X}_{i}=G(Y_{i-1})+M_{i}$ where $G(\cdot)$ represents the non-linear layers. This adds additional complications for the adversary. For instance, negative values cannot be reversed for ReLU activation layers.  .} security analysis of ShadowNet. Let us consider a CNN with $k$ standard convolutional layers where $W_i$ denotes the convolution filter for the $i$-th layer.  Additionally, $X_i$ ($Y_i$) denote the input (output) for the $i$-th layer. For a transformed filter $\hat{W}$, let $\mathcal{F}(\hat{W})$ represent the set of filters that could have been transformed to $\hat{W}$, i.e., the set of possible pre-images for $\hat{W}$. We call this the \textit{feasible set} for $\hat{W}$. The exact construction of $\mathcal{F}(\hat{W})$ is detailed in App. \ref{app:sec_ana}.  The view of the adversary is equivalent to that of the Normal World and is illustrated in Fig. \ref{fig:attack}.
\vspace{-2mm}\begin{theorem} For a CNN with $k$ convolutional layers and a given view of the normal world\footnote{For simplicity and ease  of exposition, we assume that the TEE is perfectly secure, i.e., it acts as a trusted third party, and that it has access to a true random number generator. In practice, we would have to use a pseudorandom number generator (\textsf{PRNG}) with some security parameter  $\kappa$. We can account for this by assuming a computationally bounded adversary and considering an additional $\textsf{negl}(\kappa)$ term in the Eq. \eqref{eq:theorem} where $\textsf{negl}(\cdot)$ is a negligible function.}  $\textsf{View}_{\textsf{Normal}}=\big(X_1,Y_k,\hat{W}_1,\cdots,\hat{W}_k,X'_2,\cdots,X'_k\big)$, we have:\vspace{-2mm} \begin{gather}\small\forall i \in [k] , \forall (W',W'') \in \mathcal{F}(\hat{W_i}) \times \mathcal{F}(\hat{W_i})\\\mathrm{Pr}\big[W_i=W'|\textsf{View}_{\textsf{Normal}}\big]=\mathrm{Pr}\big[W_i=W''|\textsf{View}_{\textsf{Normal}}\big] \label{eq:theorem}\end{gather}  \label{thm:1}\vspace{-2mm}\end{theorem}\vspace{-2mm}
\begin{figure}
  \centering
  \includegraphics[width=0.7\linewidth]{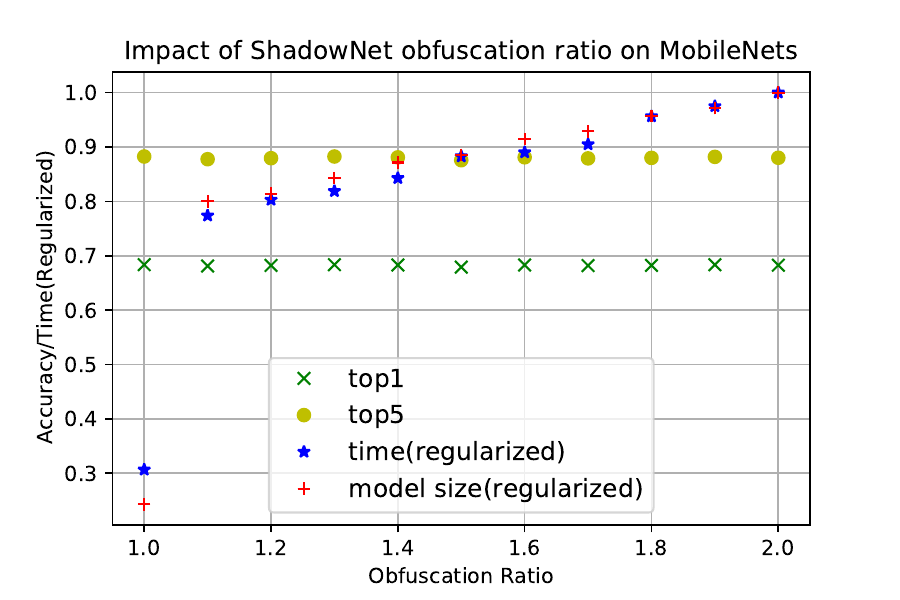}\vspace{-2mm}
    \caption{\textbf{Analysis of performance with varying obfuscation ratio.}
 We measure the change in accuracy (top1 and top5), size (regularized) and
    inference time (regularized) with the obfuscation ratio, $r$, (ranging from 1 to 2) for MobileNet. Here, $r = 1$ refers to the original model and is treated as the baseline.
    }
  \label{fig:obf_ratio} \vspace{-5mm}
\end{figure}

\input{analysis}

%% file: attack.tex
 \begin{figure}[t]
     \centering
     \begin{boxedminipage}[t]{.25\textwidth}
         \begin{center}
             \textbf{Adversary's view of ShadowNet}
         \end{center}
         \vspace{-1em}
       %  \procedure[codesize=\scriptsize]{}{
         \begin{tabbing}
             \textbf{Input:}\=~$X_1$ 
             %\\[0.1\baselineskip][\hline]
             \\[-0.75\baselineskip]
             \rule{\linewidth}{\arrayrulewidth}\\
             \textbf{View (Normal World):}\\
             \> $\hat{W}_i$, $i \in [k]$\\
             \> $ X'_{i} = Y_{i-1}+M_{i}$, $i \in [2,k]$\\
             %\\\> $I_0 = x, m_0 = input\_mask,$
             %\\\> $u_0 = \boldsymbol{0}, u_1 = input\_unmask$\\
             %\hdashline\\
             \rule{\linewidth}{\arrayrulewidth}\\
             \textbf{Output:} $Y_k$, \textit{output of the model}\\
             \rule{\linewidth}{\arrayrulewidth}\\
             \textbf{Goal of Adversary: } \textit{ Find weights of $W_i$, $i \in [k]$}
       %  }
         \end{tabbing}
     \end{boxedminipage}
      \vspace{-4pt}
      \caption[Hello]{\textbf{Formalization of the adversary's goal.} Let $\mathcal{M}(\cdot)$ denote a CNN with $k$ convolutional layers. Given input $X_1$ (input to the first layer), 
      the model's output is $Y_k$ (output of the last layer). The input of $i$-th layer is the output of the previous layer\footnotemark, i.e., $ X_i=Y_{i-1}, i \in [2,k]$. For ShadowNet, $X'_{i}$ denotes the masked input for the $i$-th layer which is given $Y_{i-1}+M_i$.
      For each convolutional layer, the adversary can observe the masked input $X'_i$ and the  transformed filter $\hat{W}_i$.
     The adversary's goal is to find the original weights of $W_i$.}
      \label{fig:attack} \vspace{-5mm}
  \end{figure}

%% file: analysis.tex
\begin{sketch} Every input/output pair for the intermediate layers $i\in [2,k-1]$ is embedded in $\F$ and masked. As a result, $X'_i$ is indistinguishable from a randomly chosen input of the same shape in $\F$. Hence, $X'_i$s clearly contain no information about the convolution filters. For the first (last) layer, the output (input) is masked which also prevents any reconstruction of the true weights for $W_1$ ($W_k$). The rest of the proof follows trivially from the construction of the feasible sets $\mathcal{F}(W_i)$. The full proof is presented in App. \ref{app:sec_ana}.
\end{sketch}

The above theorem states that, on observing a transformed filter $\hat{W}_i$, corresponding to any convolutional layer $i$, an adversary cannot distinguish between two filters that belong to its \textit{feasible set}. Thus, the feasible sets act as cloaking regions for the original weights. Intuitively, the parameter $r$ is analogous to the security parameter (such as, key size) for generic cryptographic protocols. The larger the value of the obfuscation ratio $r$, the greater is the size of the feasible set and consequently, the better is the security. Concretely, $\mathcal{F}(\hat{W}_a)\supset \mathcal{F}(\hat{W}_b)$ where $m_a=|\hat{W}_a|>|\hat{W}_b|=m_b$ (equivalently, $r_a>r_b$). $r$ is essentially a trade-off between the security guarantee and performance. The exact size of feasible sets can be analytically computed (see App. \ref{app:sec_ana}). Fig. \ref{fig:obf_ratio} shows ShadowNet’s performance with varying values of $r$. Based on this, we set $r=1.2$ for our experimental setup since this was the sweet spot. Specifically, even with $n=16$ (smallest $n$ for our evaluation) and field $\mathbb{F}=\mathbb{Z}_p$
  for $p=2^{24}-3$, the size of the feasible set, $|\mathcal{F}(\hat{W})|$, is of the order of $2^{268}$
  which is sufficient for security. The reason why we get a large feasible set even with a relatively small value of $r$ is that the random permutation matrix $P_\pi$
  and random scalar multiplications by $\lambda_i$s also contribute significantly to the size of the feasible set. %In fact, if we quantize all the inputs and weights to integers, embed them in the field $\mathbb{Z}_p$ for a sufficiently large prime $p$  \cite{tramer2018slalom} and draw our random masks from $\mathbb{Z}_p$, then for $r=2$, the feasibility set for any transformed filter, $\hat{W}$ is essentially the set of all possible filters. This perfectly hides all the true kernels. A detailed discussion is provided in Appendix \ref{app:sec_ana}.

\noindent \textbf{Note.} In the above guarantee, the formal security of the mask layers is rooted in the quantization step. Specifically,  the quantization operation embeds the value in a finite field, $\mathbb{F}$. Subsequently, the masks can be chosen uniformly at random from $\mathbb{F}$ thereby making the mask layer equivalent to one-time-pad encryption \cite{doi:10.1080/0161-119691885040}. %Without quantization, we would have to deal with real numbers (floating points) -- we cannot give any formal security guarantee here since there exists no uniform distribution over reals.}
 
Based on Thm. \ref{thm:1}, we present the following conjecture:\vspace{-2mm}
\begin{conj}Let $q$ be the number of queries required for a  model stealing attack with access to just the querying API, i.e, $Y=\mathcal{M}(X)$\footnote{This corresponds to access to $(X_1,Y_k)$ from Fig. \ref{fig:attack}.} (black-box model stealing attack) and some information about the model architecture (such as, number and type of convolutional layers). Let $q'$ be the number of queries required for an attack on  ShadowNet. We conjecture that $q'$ is of the order of $O(q)$.\label{conj:1}\end{conj}

Black-box model stealing attacks can be typically classified into two types -- $(1)$ functionally-equivalent model stealing attacks \cite{jagielski2019high} where the attacker tries to extract the \textit{exact} weights of model by analytically solving a set of linear equations, $(2)$ learning-based model stealing attacks where the adversary tries to learn a shadow model \cite{tramer2016stealing,papernot2017practical,Knockoff,MixMatch}.  The first type of attacks are \textit{impossible} in ShadowNet -- since our mask layers are refreshed every time, it is mathematically impossible to solve for the model weights (Lemma \ref{lemma:1}, App. \ref{app:sec_ana}). Clearly from our threat model (Sec. \ref{sec:threat}), protection against learning-based black-box model stealing attacks is \textit{beyond the scope} of ShadowNet. The implications of Conjecture 1 in this context is that with ShadowNet, an adversary cannot do anything significantly better than a standard learning-based black-box\footnote{with some extra information about the architecture as stated in Conj. \ref{conj:1}} model stealing attack.
Our reasoning is based on the fact the feasible sets for the transformed weights are sufficiently large -- this provides sufficient cloaking region for the original weights. In other words, an adversary \textit{cannot} learn anything useful about the original weights by observing the ShadowNet transformed weights. We provide empirical evidence in support of our conjecture as follows. 

%\zs{As we have evaluated on three types of BB attack, do we still need the Illustration here? Should the experiments results speak for themselves? My concern is that the current Illustration based on some ad-hoc model parameters.}

%defined in Section~\ref{sec:transform}. 
%It can be equal or bigger than 64, namely the number of filters in the original \textit{conv} layer.
%Hence, we conjecture that the adversarygains no advantage  by  training an equivalent CNN with the transformed weights from ShadowNet.}

\point{Empirical Analysis}
Table \ref{tab:BB} shows the empirical evaluation of our conjecture based on two black-box attacks, namely, Knockoff attack ~\cite{Knockoff} and MixMatch~\cite{MixMatch,jagielski2019high}, on a victim model of a four layer fully-connected CNN trained on CIFAR-10 (Fig. \ref{fig:example-cnn}).  Knockoff is the state-of-the-art black-box attack with an \textit{adaptive} query strategy.  A recent survey~\cite{pow2022iclr} shows that MixMatch is currently the state-of-the-art attack with the \textit{highest} attack accuracy. Knockoff queries the victim model from an out-of-distribution dataset. For our evaluation,  Knockoff uses 10K queries from CIFAR-100~\cite{cifar100} selected via an adaptive strategy based on reinforcement learning. MixMatch is based on semi-supervised learning; the attack samples 8K points from CIFAR-10 for querying the victim model and uses an additional 8K unlabeled points from CIFAR-10 for its semi-supervised training. 
 For Knockoff, the victim model is trained on CIFAR-10 with accuracy 81.6\% (column 1 in Table \ref{tab:BB}). For MixMatch, the victim model is trained on CIFAR-10 with the MixMatch semi-supervised training approach and has accuracy 98.1\% (column 1). We used different training strategies for the victim model for the two attacks in order to be consistent with the configurations of their respective original papers~\cite{Pow_repo}. 
The baseline shadow model (column 2) for  Knockoff and MixMatch is and ResNet-18 and Wide-ResNet-28, respectively.

In order to assess whether the adversary benefits from knowing the transformed weights ($\hat{W}$) in ShadowNet, $(1)$ we create a custom adversary model with the 
same architecture as that of the victim model, $(2)$ we copy weights
$\hat{W}$ to the adversary’s model, $(3)$  we mark the convolutional layers with weights $\hat{W}$ as non-trainable and 
train only the other layers (results in column 3). The resulting adversary model is depicted in Fig. \ref{fig:attacker_net} in App. \ref{app:sec_ana}.
 We observe that the ShadowNet adapted attack  is less accurate than the baseline attack (column 2) after the same number of queries. In fact, it performs worse than a random baseline (column 4) where the weights of the non-trainable convolutional layers (as described above) are randomly assigned. This shows that the ShadowNet transformed weights contain no useful information about the original weights and the adversary gains \textit{no advantage} by reusing the transformed weights.   
%This backs up our conjecture that ShadowNet does not give the adversary any advantage even with a black-box attack.
\begin{table}[!ht] \vspace{-2mm}
\centering
\caption{Empirical evaluation of black-box attack}
\label{tab:BB}
\resizebox{\columnwidth}{!}{
\begin{tabular}{c|c|c|c|c}\hline
& Victim & Attack Baseline & ShadowNet & Random Baseline  \\\hline
% we don't have consistency data for MixMatch
%\new{Consistency}&\new{40.1\%} & \new{33.1\%}& \new{38.4\%} \\
Knockoff~\cite{Knockoff} & 81.6\%& 36.6\% & 31.1\% & 36.2\%  \\
MixMatch~\cite{MixMatch,jagielski2019high} & 98.1\% & 95.8\% & 92.3\% & 92.4\%  \\
\hline
\end{tabular}}
\footnotesize
    {  Column 1 reports the underlying victim model's accuracy. Column 2 corresponds to the attack baseline where the adversary has no knowledge of the victim model. Column 3 corresponds to the attack setting adapted for ShadowNet where the adversary has access to the victim model's transformed weights. Column 4 corresponds to the case where the convolutional layers of the attacker's model have random weights. }
    %In order to assess whether the adversary benefits from knowing the transformed weights (conv') in ShadowNet, (1) we create a customized adversary model with the same architecture as that of the victim model, (2) we copy theconv’ layers to the adversary’s model, (3)  we mark these conv’ layers as non-trainable and onlytrain the other layers (results in column 2). For the Random Baseline attack,  the weights of the non-trainable convolutional layers are randomly initialized. } 
    \vspace{-4mm}
    \end{table}

%% file: discussion.tex
\section{Discussion}\label{sec:discussion}
%\vspace{-2mm}
\point{Running unmodified model inside the TEE} There are several challenges in running the unmodified model inside the TEE. First, standard model inference (outside the TEE) is memory intensive relative to the original model size. For instance, the original model size for MiniVGG and ResNet-44 is 20MB and 2.4MB, respectively. However, running the unmodified model inference requires dynamically allocating at least 44MB (2.2X) and 10MB (4.2X) for MiniVGG and ResNet-44, respectively. This is because running convolutional layers for standard model inference requires reshaping a 3D weight matrix into 2D matrices – this needs large chunks of intermediate dynamic memory allocation to store the extra copy of the resized weights. In contrast, ShadowNet only requires 4MB (0.2X) and 5MB (2.1X) TEE memory for MiniVGG and ResNet-44, respectively. ShadowNet reduces the memory footprint by keeping the memory-intensive linear layers outside the TEE. 
ShadowNet’s memory requirement for MobileNets (original model size 17.2MB) is relatively high (48MB -2.7X) since the MobileNets architecture introduces special CNN layers such as depthwise and pointwise convolution which increases the size of the input and output for each layer. As a result, the mask layer added by ShadowNet consumes a relatively larger amount of TEE memory. Nevertheless, it is still significantly lower than running the unmodified model inside TEE (61MB-3.5X).
%For instance, convolution operations are expensive as they require resizing of the dimensions of the matrices. Additionally, buffers are allocated dynamically to store the intermediate results for the output of each layer. Our memory profiling shows that running the unmodified model inference requires dynamically allocating at least 44MB for MiniVGG, 10MB for ResNet-44, 61MB for MobileNet, 54MB for ResNet-404 and 72MB for YOLOv4-tiny. Also note that the secure OS inside the TEE reserves $>$4MB memory for its usage. Hence given the TEE memory limit of $64$ MB for our evaluation board, it is not feasible to run the entire model inference task inside the TEE. ShadowNet, on the other hand, only requires $4$MB for MiniVGG, $5$MB for ResNet-44 and $48$MB for MobileNet,  $34$MB for ResNet-404, $52$MB for YOLOv4-tiny, thereby requiring a much smaller secure memory size. 
Second, running unmodified models inside the TEE would require a significant engineering effort. TensorFlow Lite needs to be ported into the TEE to support the convolutional layers and matrix operations which rely on certain mathematical libraries (in C++) for efficiency. This entails a significant engineering effort since currently Arm OP-TEE OS does not support C++ and its associated computing libraries.  Third, the limited size of the TCB presents additional challenges. TensorFlow Lite has tens of thousands lines of code and porting it as a whole inside TEE would require a much larger TCB. On the other hand, recall that ShadowNet can work with a small TCB -- its TA adds only 2100 LOC inside the TEE (the new operations add $<200$ LOC inside the TCB). Additionally, note that the extra code added for model conversion is used offline by the model vendor and is not deployed on the device. Hence, this does not increase the on-device TCB size. Last, we will lose access to hardware accelerators if we run the unmodified model inside the TEE.

  \begin{figure}
   \centering
   \includegraphics[width=0.6\linewidth]{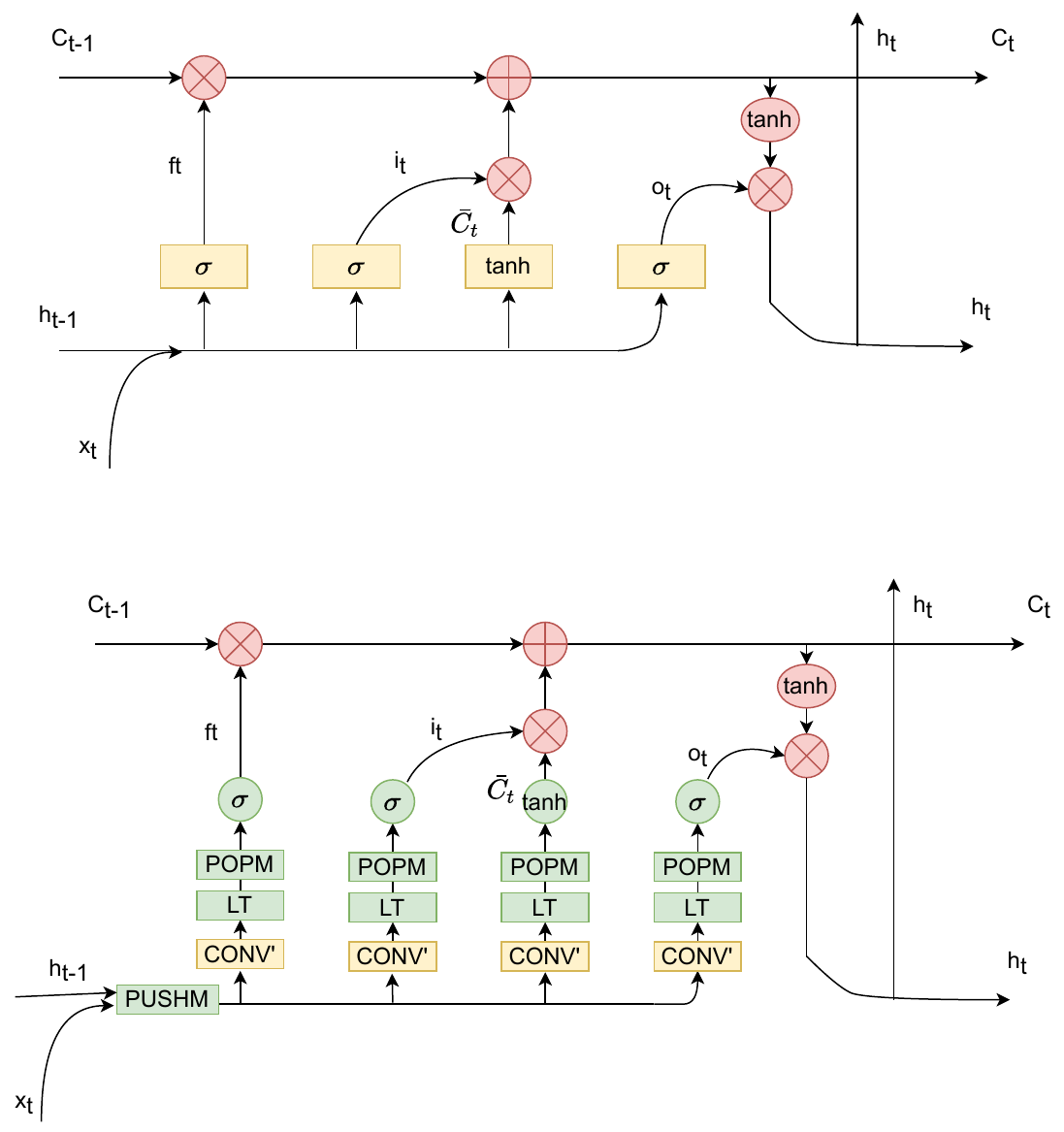}\vspace{-2mm}
     \caption{\textbf{Proof-of-concept transformation on LSTM}. The original LSTM layer is showed
     on the top and the ShadowNet transformed LSTM is depicted at the bottom. The red circles represent
     pointwise matrix operations and the yellow square boxes marked with activation symbols represent
     fully connected layers followed by the corresponding activation layers.  Green boxes marked ``LT", ``PUSHM" and ``POPM" represent  the ``LinearTransform", ``PushMask" and ``PopMask" layer of ShadowNet, respectively. The yellow square boxes marked with ``CONV$'$" represent the ShadowNet transformed convolutional layers
     which are outsourced to the untrusted world. 
     All the other parts of the LSTM layer including the internal states remain in the secure world.
     }
   \label{fig:lstm}\vspace{-6mm}
 \end{figure}
\textbf{Support for CNNs and LSTMs.} %Our evaluation (Sec. \ref{sec:eval}) covers a diverse range of CNN architectures. %Following the same transformation rules as showcased for , ShadowNet can be applied on InceptionNet, DenseNet and so on.  
In general, ShadowNet can be applied to any CNN model as long as the TEE can support the memory requirements of the corresponding ShadowNet transformed model. The amount of TEE memory required for running the ShadowNet transformed model can be estimated by the size of the weights that need to be stored inside the TEE (\textit{Other} in Table \ref{tab:sd-size}). Concretely, the total size is given by the sum of the size of the weights of four ShadowNet operations, namely, \textit{LinearTransform}, \textit{ShuffleChannel}, \textit{AddMask} and \textit{Batchnorm}.  The size of \textit{LinearTransform} and \textit{ShuffleChannel} can be computed from the shapes of the convolutional layer and the obfuscation ratio. The size of \textit{AddMask} can be computed from the input and output shape of each layer. The size of \textit{Batchnorm} can be computed directly from the model. %As shown in Table \ref{tab:sd-size}, even though the original model can be big, the weights need to be stored inside TEE can be small (e.g., MiniVGG) or big (e.g., MobileNets). 
%Note that the actual run-time memory allocation required by ShadowNet is lesser due to our optimized memory allocation (Sec. \ref{sec:ca-ta}). }
%The CA/TA might be a little different. For example, for networks with shortcuts,
%like DenseNet and , the TA will need more TEE memory to keep the previous 
%layers' output as input for the following layers.
\\ShadowNet can also be used for LSTMs. LSTMs typically consist of fully-connected layers, pointwise matrix operations and activation layers. Recall that a fully connected layer can
     be treated as a convolutional layer during the ShadowNet transformation. Hence for LSTMs, ShadowNet applies linear transformation on the fully connected layers while keeping the pointwise matrix operations and activation layers inside the TEE. As a concrete example, a LSTM layer (top)  and its corresponding ShadowNet transformation (bottom) is depicted in Fig. \ref{fig:lstm}. We implement a prototype ShadowNet transformation for the LSTM layer depicted in Fig. \ref{fig:lstm} with an input ($x_t$) shape  of $(12, 30)$, output space/units set to be $10$, obfuscation ration $r=1.2$ and $1640$ parameters.  %The ShadowNet transformed layer has the same input and output shape as the original LSTM layer and can seamlessly replace it. 
    The number of parameters after the ShadowNet transformation is $3640$ -- the extra parameters are due to the mask layers and the transformed CONV$'$ layers.
     %We also performed additional experiments on the performance of the LSTM model depicted in Fig. 10.
%LSTM Model details:
%Number of features = 12
%Number of dimensions = 30
%Number of time steps = 10
%Obfuscation Ratio = 1.2
Original LSTM model inference time is 33 ms while the
Shadownet transformed LSTM model inference time is 41ms (1.24X).
%If the weights are changing dynamically, then the transformation can not be performed off-line. Hence, we need to move the pointwise operations on cell states to the TEE as it changes over time.
\\\textbf{Support for Cloud Platform.}
ShadowNet  can be used for secure model inference in the cloud as well. For this, 
the ShadowNet CA/TA needs to be changed to support a cloud TEE, such as SGX. 
Compared with other designs that perform model inference inside SGX~\cite{kunkel2019tensorscone,tftrusted}, we expect ShadowNet to be more 
efficient in using the SGX memory and to benefit from the co-located GPUs for acceleration.
\\\textbf{Layerwise ShadowNet.}
ShadowNet transformations can be applied on each convolutional layer independently. Hence, an alternative strategy for implementation is to selectively apply ShadowNet to only the sensitive layers which would improve performance. The rationale behind Layerwise ShadowNet is supported by research on transferable learning~\cite{yosinski2014transferable} 
which shows that the bottom layers contain features that are more specific to the training dataset. Hence, these features are more sensitive than the generalized features in the top layers.

%% file: relatedwork.tex
\section{Related Work}\label{sec:related_work}
%\vspace{-2mm}
%In this section, we give an overview of the related prior work.
%With on-device machine learning(ML) becoming a new trend\cite{xu2019first}, 
%thousands of proprietary ML models stored on untrusted mobile devices are at risk of leakage\cite{sun2020mind}.
Existing research on secure ML covers both end devices and cloud-based solutions.
Offline Model Guard (OMG) \cite{bayerl2020offline} provides a secure model inference framework for mobile devices
based on SANCTUARY~\cite{brasser2019sanctuary}, a user space enclave built on Arm TrustZone. However, the original paper presents only a proof-of-concept implementation for OMG that cannot be directly integrated with existing mobile apps (unlike ShadowNet). %Additionally, it is evaluated only on a single voice recognition model.
Moreover, OMG is based on the Sanctuary enclave~\cite{brasser2019sanctuary} which runs a user application along with the OS in an isolated environment. For model inference, the unmodified model is run inside the Sanctuary enclave. Thus, the threat model is different as OMG requires the whole software stack, including the OS, libraries and the application inside Sanctuary, to be trusted. ShadowNet, on the other hand, relies only on a secure TEE-OS and TA which is a more relaxed trust assumption.
%OMG allows the model inference framework to run fully inside the SANCTUARY enclave to protect the model privacy.
MLCapsule \cite{hanzlik2018mlcapsule} also deploys the model on the client side to protect the user input
from being sent to the untrusted cloud. Additionally, it runs the model inference inside SGX
to prevent the model from being leaked to the client.
DarkneTZ~\cite{darknetz} is a secure machine learning framework built on top of Arm TrustZone. 
It allows a few selected layers to run inside the TEE to protect part of the model.
OMG, MLCapsule and DarknetTZ do not support secure GPU acceleration and have a larger TCB size than ShadowNet.
%ShadowNet has a comparatively smaller TCB size and allows secure outsourcing of the linear layers to the GPU.
Graviton \cite{volos2018graviton} proposes TEE extension for GPU hardware, thus allowing GPU tasks
to run securely.
%Graviton builds encrypted channels for CPU’s TEE (like Intel SGX or Arm TrustZone) 
%to talk to GPU, and enforces memory mapping check for GPU
%commands so that different GPU tasks are isolated. 
 However, it requires hardware changes to the GPU.
Secloak~\cite{lentz2018secloak} partitions the GPU into Secure World to run GPU tasks securely at a high performance
penalty. ShadowNet does not change the GPU hardware or partition the GPU into the Secure World.

%Research on secure ML in the cloud is an active area.
TensorSCONE~\cite{kunkel2019tensorscone} proposes a secure ML framework that runs on the untrusted cloud. %TensorSCONEintegrates TensorFlow with the secure Linux container technology, SCONE~\cite{arnautov2016scone}, guarded by SGX.  
However, it is only evaluated on Inceptionv4 and has a 3.1X time overhead with 330MB memory consumption -- these overheads are higher than that of ShadowNet.
Additionally, it is designed for a cloud environment while ShadowNet is aimed at a mobile environment which is more resource constrained. 
TF Trusted \cite{tftrusted} leverages custom operations to send gRPC messages to the 
Intel SGX device via Google Asylo~\cite{asylo}.% where the model is then run by Tensorflow Lite. 
%Occlumency \cite{lee2019occlumency} provides a suite of heuristic techniques based on Caffe and improves inference speed by 3.6 times.
%Running model inference inside TEE faces challenges due to limited memory and lack of GPU acceleration.
This requires more TEE resources than ShadowNet and does not support GPU acceleration.

Slalom \cite{tramer2018slalom} outsources the
linear layers to the GPU for acceleration with masked inputs while keeping the other layers inside SGX.
%It verifies the linear layer's results and computes the nonlinear layers inside the SGX. 
 Slalom protects the user input privacy but not the model weights\footnote{Slalom outlines a conceptual way for protecting the model privacy from clients -- however, no concrete implementation and evaluation is provided.} from the 
untrusted server while ShadowNet protects the model weights.
YerbaBuena \cite{gu2018securing} partitions the model into frontnets %(such as the first layer) 
and backnets, and executes the frontnets
inside SGX. This protects the input from the cloud.
SecureNets~\cite{chen2018securenets} transforms both the input and the linear layer into matrices
and applies matrix transformations~\cite{salinas2016efficient}
before sending them to the cloud. 
It is unclear whether SecureNets supports depthwise convolution and convolution with stride.
ShadowNet does not require transforming input and weights into matrices and is compatible with existing linear operations.

Some secure ML systems use cryptographic primitives, such as CryptoNets~\cite{gilad2016cryptonets}, \cite{jiang2018secure}, TF
Encrypted \cite{tfencrypted} and SafetyNets~\cite{ghodsi2017safetynets}.
ShadowNet's performance is orders of magnitude better than such cryptographic approaches. For instance, for a single image classification, CryptoNets takes $\sim$570s s on a PC while ShadowNet takes $<$1s on a smartphone.

%CryptoNets~\cite{gilad2016cryptonets} applies HE on neural networks and runs model inference on the encrypted data to protect user input privacy.Jiang et. al \cite{jiang2018secure} presents a solution for encrypting a matrix using HE to protect both the user data and the model. TF Encrypted \cite{tfencrypted} enables training and prediction over encrypted data via MPC and HE. SafetyNets~\cite{ghodsi2017safetynets} designs an interactive protocol that allows clients to verifythe correctness of a class of DNNs running on the untrusted cloud. 

%To ensure the integrity of model weights, Uchida et. al \cite{uchida2017embedding} and Zhang et. al \cite{zhang2018protecting} embed watermarks into deep neural model parameters while training the models.
%DeepAttest \cite{chen2019deepattest} encodes fingerprint in DNN weights to prevent weight modification. These works do not prevent weight leakage.

%% file: conclusion.tex
\vspace{-0.5mm}\section{Conclusion}
In this paper, we have proposed ShadowNet, a secure on-device model inference system for CNNs
that protects the model privacy with a TEE while leveraging the untrusted 
hardware for acceleration. %The key idea of ShadowNet is to  apply linear transformation on the weights of the linear layers and outsource them to the untrusted world. %In this way, ShadowNet leverages the hardware accelerators without trusting them and restores the results inside the TEE. 
We have implemented an end-to-end prototype of ShadowNet on TensorFlow Lite
and OP-TEE with optimizations to work with a small TCB. %We also provide automated model conversion tool to easy adoption.% the developer's effort in adopting ShadowNet system. %Our evaluation on three popular CNNs, namely, MobileNets, ResNet-44 and MiniVGG demonstrates ShadowNet's practical feasibility on real-world datasets.  

%% file: appendix.tex
%\clearpage
\begin{appendices}

\section{Optimizing the \text{sqrt} function} 
    \label{sec:sqrt}

\begin{table}[ht]
\centering
\caption{Performance of different \textit{sqrt} implementation.}
\label{tab:sqrt}
\begin{tabular}{c|c|c|c|c}
\hline
    \textit{Sqrt} Impl. & Time(ms) & S/H & Algorithm & CFLAG\\
\hline
    GNU libc & 3.53  & S& IEEE754 & Default\\
    Newlib & 13.78  & S& IEEE754 & \textit{-O2} \\
    Our TA & 194.04  & S& Newton & \textit{-Os}\\
    Arm VFP & 10.86 & H& unknown & Default\\
    Arm Neon & 6.62 & H& Newton& Default\\
\hline
\end{tabular}
\footnotesize
{\\
\raggedright \textit{Note}: 
    a. S/H: \textit{S} means Software based implementation, \textit{H} means Hardware based implementation, like special instructions;
    b. CFLAG: GCC compilation flag; c. IEEE754 means algorithm exploits bits hacking of IEEE754 float format; d. Newton means Newton 
    Iteration for sqrt.}
\end{table}

%We highly optimize functions that incur big overhead. One example is Batchnorm layer's 
%\textit{sqrt} function. 
%We will use it as an example to show that finding an efficient implementation
%for TA is non-trivial. 
There are many different implementations of the \textit{sqrt} function for floating point numbers
for AArch64 architecture.
   % including 
%software-based and hardware-based implementations.
Software-based implementations include algorithms using Newton iteration and 
bits hacking of IEEE754 float representation. Hardware-based implementations include 
Arm VFP support for \textit{fsqrt}, and Arm Neon support for float \textit{sqrt}.
Additionally, the performance of software-based implementations is also affected by the compilation flag. The default
gcc compilation flag for TA is \textit{-Os}, which optimizes space first; if we change it to \textit{-O2}, the
performance is more than 100x faster while the TA size increases from 55KB to 67KB.
We evaluate all the above implementations by doing 3,200,000 sqrt operations and show the results in Table~\ref{tab:sqrt}. 
Our TA initially used a software implementation using Newton Iteration algorithm. After evaluation, 
we switched to the Arm Neon based \textit{sqrt} implementation for the speed and ease of implementation.
%Although it is the second fastest, it is easier to port to TEE than GNU libc's implementation.
% Ruimin: {add the conclusion of the results, such as We chose xxx for its fastest computaion.}

%\section{Handling shortcuts in CNNs}
%     \label{sec:resnet}

%\section{Adversary's Equivalent CNN}
%     \label{sec:attack-net}

\input{Appendix_sec}

\end{appendices}

%% file: Appendix_sec.tex
\section{Security Analysis}\label{app:sec_ana}

\begin{figure*}[ht]
   \centering
   \includegraphics[width=0.70\linewidth]{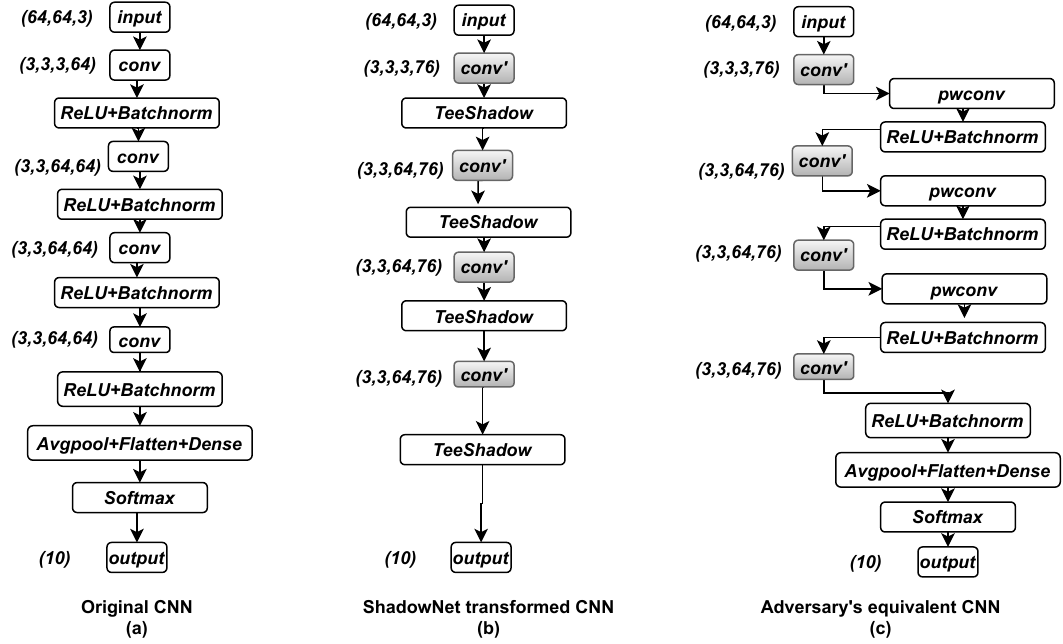}
     \caption{Figure (a) shows the original CNN model. Figure (b) is the ShadowNet transformed model. Figure (c) the adversary's equivalent CNN. Before training, its \textit{conv'} layers are initialized with weights copied from the corresponding \textit{conv'} layers of the ShadowNet transformed model. These \textit{conv'} layers are then reused by the adversary and marked as non-trainable during training. }
   \label{fig:attacker_net}
 \end{figure*}
\point{Construction of Feasible Set}
We refer to the kernels $f_i$ in the random filter $F=[f_1,\cdots,f_{m-n}]$ as mask kernels. Additionally, let $\in_R$ represent a uniform random sampling. In what follows, we outline the methodology to compute the feasible set, $\mathcal{F}(\hat{W})$ for a  given transformed weight matrix $\hat{W}$. The idea is to back-trace and compute the set of possible pre-images.
Now, the feasible set is constructed as follows: \begin{enumerate}
    \item Select the set of $m-n$ indices uniformly at random:  \begin{gather} \Omega \subset_R [m], |\Omega|=m-n \end{gather}
    $\Omega$ represents a possible set of indices that correspond to the mask kernels. 
    \item The corresponding set of mask kernels is: 
    \begin{gather} \Phi_\Omega=\{\hat{w}_i| i \in \Omega\} \end{gather} 
    \item Let $\overline{\Phi}_{\Omega}=\{\hat{w}_i|i \in [n]\setminus \Omega\}$ be the set of  transformed original kernels. Additionally, let $\overline{W}=[\overline{w}_1,\cdots,\overline{w}_n]$ where $\overline{w}_i \in \overline{\Phi}_{\Omega}$ and $ \overline{w}_i\neq \overline{w}_j, i,j \in [n], i \neq j$.
    \item Sample a random permutation $\sigma \in_R S_n$. We assume that $\sigma=\pi[1:n]^{-1}$, i.e,, $\sigma$ reverses the effect of $\pi$ on the transformed original kernels. Thus, $\overline{W}_{\sigma}=[\overline{w}_{\sigma(1)},\cdots, \overline{w}_{\sigma(n)}]$ represents a possible transformed filter.
    \item Compute
    \begin{gather*}\forall i \in [n]\\ \mathcal{F}^i_{\Omega,\sigma}(\hat{W}) =\{w|w=d\cdot (\overline{w}_{\sigma(i)}-\hat{w}'), d \in_R \F, \hat{w}'\in_R \Phi_\Omega \}\end{gather*}
    $\mathcal{F}^i_{\Omega,\sigma}(\hat{W})$ represents the set of possible values for the kernel $w_i$ for the given  $\Omega$ and $\sigma$.
    \item Compute 
    \begin{gather*}\mathcal{F}_{\Omega,\sigma}(\hat{W})=\big\{\big[w_1,\cdots,w_n\big]\big |\forall i \in [n], w_i \in \mathcal{F}_{\Omega,\sigma}^i(\hat{W})\big\} \big\}\end{gather*}
    $\mathcal{F}_{\Omega,\sigma}(\hat{W})$ denotes the set of possible filters $W$ for the given $\Omega$ and $\sigma$.
    
    \item Clearly, we have 
    \begin{gather}\mathcal{F}(\hat{W})=\bigcup_{\Omega}\bigcup_{\sigma}\mathcal{F}_{\Omega,\sigma}(\hat{W})\numberthis \label{Eq:1}\end{gather}
\end{enumerate} 
Clearly, larger the value of $r$, greater is the size of $\Omega$ and consequently, $\mathcal{F}(\hat{W})$. Additionally, it is evident that $\mathcal{F}(\hat{W}_a)\supset \mathcal{F}(\hat{W}_b)$ where $m_a=|\hat{W}_a|>|\hat{W}_b|=m_b$ (equivalently, $r_a>r_b$).
For depthwise convolutional layer, we have \begin{gather}\mathcal{F}(\hat{W})=\{ [d_1\cdot \hat{w}_{\sigma(1)},\cdots, d_n\cdot \hat{w}_{\sigma(n)} ]\big|\sigma \in_R S_n, \forall i \in [n]~d_i \in_R \F\}\label{Eq:2}\end{gather}

\textbf{Theorem 1.}\textit{ For a CNN with $k$ convolutional layers and a given view of the normal world} $\textsf{View}_{\textsf{Normal}}=\big(X_1,Y_k,\hat{W}_1,\cdots,\hat{W}_k,X'_2,\cdots,X'_k\big)$\textit{, we have:} \begin{gather*}\forall i \in [k] , \forall (W',W'') \in \mathcal{F}(\hat{W_i}) \times \mathcal{F}(\hat{W_i})\\\mathrm{Pr}\big[W_i=W'|\textsf{View}_{\textsf{Normal}}\big]=\mathrm{Pr}\big[W_i=W''|\textsf{View}_{\textsf{Normal}}\big]\numberthis \label{Eq:4} \end{gather*} 
\\\\
\begin{proof} First, we present two helper lemmas as follows.
\begin{lemma}$X'_i, i \in [k]$ is indistinguishable from a random tensor in $\F$ with the same shape as $X_i$.\end{lemma}
\begin{proof}Note that $X_i$s are embedded in a field $\F$. Thus clearly, masking the inputs $X'_i=X_i+M$ is equivalent to applying a one-time pad which concludes our proof.\end{proof}
\begin{lemma}Adversary cannot reconstruct $W_i, i \in [n]$ from  $(X_1,Y_k)$.  \label{lemma:1}\end{lemma}
\begin{proof}Recall that the goal of the adversary is to find $W_i$, i.e., solve for the $|W_i|$ variables.

 %Hence, $X'_i$ is indistinguishable from any random input in $\F$ with the same shape as $X_i$. Clearly, which contains no information about $W_i$.   Now, 
 Consider the first layer -- the adversary has access to the true input $X_1$ but only gets to see the masked output $X'_2=Y_1+M_2$. In other words, adversary has $2|Y_1|+|W_1|$ unknown variables. Consequently, the adversary cannot solve\footnote{The input/output pair $(X_i,Y_i)$ for any layer is connected to $W_i$ by a system of linear equations. Hence, an adversary needs access to \textit{both} the input and the output to solve for $W_i$} for the weights of $W_1$ from this. Similarly, for the last layer the adversary cannot solve for $W_k$ from $(\hat{X}_k,Y_k)$. %Now, for the other intermediate layers, both the input ($\hat{X}_i, i \in [2,k-1]$) and the output ($\hat{X}_{i+1}=Y_i+M_{i+1}$) is masked  which clearly prevents solving for $W_i$. 
 This concludes our proof for the above lemma. \end{proof}

Note that here we assume the worst case situation for ShadowNet where $X_i=Y_{i-1}$. In practice, the adversary can only observe $\hat{X}_{i}=G(Y_{i-1})+M_{i}$ where $G(\cdot)$ represents the non-linear layers. This adds additional complications for the adversary. For instance, negative values cannot be reversed for ReLU activation layers. 
%\begin{proof} Recall that the goal of the adversary is to find $W_i$, i.e., solve for the $|W_i|$ variables.  We assume that $t$ is sufficiently large and that $t$ is unknown to the adversary. Clearly, it is impossible for the adversary to unmask $\hat{X}_i$ and figure out the exact values of $X_i$. Now, consider the first layer -- the adversary has access to the true input $X_1$ but only gets to see the masked output $\hat{X}_2=Y_1+M_2$. In other words, adversary has $2|Y_i|+|W_i|$ unknown variables. Consequently, the adversary cannot solve\footnote{The input/output pair $(X_i,Y_i)$ for any layer is connected to $W_i$ by a system of linear equations. Hence, an adversary needs access to \textit{both} the input and the output to solve for $W_i$} for the weights of $W_1$ from this. Similarly, for the last layer the adversary cannot solve for $W_k$ from $(\hat{X}_k,Y_k)$. Now, for the other intermediate layers, both the input ($\hat{X}_i, i \in [2,k-1]$) and the output ($\hat{X}_{i+1}=Y_i+M_{i+1}$) is masked  which clearly prevents solving for $W_i$. This concludes our proof for the above lemma. Note that here we assume the worst case situation for ShadowNet where $X_i=Y_{i-1}$. In practice, the adversary can only observe $\hat{X}_{i}=G(Y_{i-1})+M_{i}$ where $G(\cdot)$ represents the non-linear layers. This adds additional complications for the adversary. For instance, negative values cannot be reversed for ReLU activation layers.  \end{proof}

Clearly, from our construction of the feasible set in Equations \eqref{Eq:1} and \eqref{Eq:2}, we have \begin{gather*} \forall i \in [k] , \forall (W',W'') \in \mathcal{F}(\hat{W_i}) \times \mathcal{F}(\hat{W_i})\\\mathrm{Pr}\big[W_i=W'|\hat{W}_i\big]=\mathrm{Pr}\big[W_i=W''|\hat{W}_i\big]\numberthis \label{Eq:3} \end{gather*}
Equation \eqref{Eq:4} follows directly from Lemma \ref{lemma:1} and Equation \eqref{Eq:3}, concluding our proof.

\end{proof}

\vspace{-2mm}
\begin{figure}[ht]
  \centering
  \includegraphics[width=0.7\linewidth]{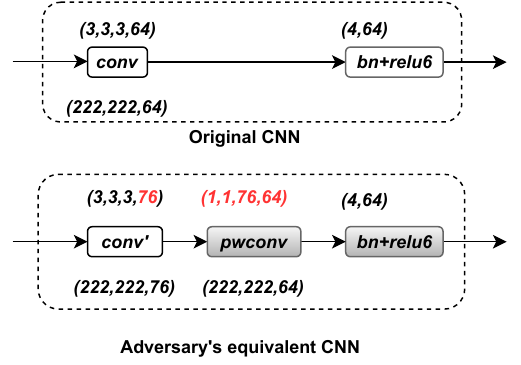}\vspace{-2mm}
    \caption{\textbf{The equivalent CNN architecture needed to be trained by an adversary.} 
    %The upper part shows the original CNN architecture 
    %(only one block of \textit{conv}+\textit{bn+relu6}, omitting the other repeating blocks for simplicity.), the bottom part
    %shows the equivalent CNN architecture that the attacker can build with the exposed weights of \textit{conv'}. 
    The grey color marks the weights of the layers with unknown parameters that the adversary has to train.
    The weights shape is marked on top of the box and the output shape is marked under the box. %Note that \textit{pwconv}  corresponds to the linear transformation performed to restore the weights in ShadowNet which is equivalent to a simple pointwise convolution.
    }
  \label{fig:equivalent} \vspace{-4mm}
\end{figure}
\point{Illustration} In order to study the advantage an adversary might have in ShadownNet, over a black-box model stealing attack, we assume that the adversary reuses the transformed weights to build
an equivalent CNN. 
Consider the example CNN from Figure~\ref{fig:example} -- the minimum equivalent CNN that 
reuses the transformed weights is shown in Figure~\ref{fig:equivalent}. 
As mentioned before in Section~\ref{sec:scheme_opt},
the linear transformation layer is essentially a pointwise convolutional layer and we use \textit{pwconv} to represent it inside the TEE. Note that the mask/unmask layers are not needed to construct the equivalent CNN.

We use the number of learnable parameters to assess the difficulty of training a CNN.
In our example, the block in the original CNN has $3\times3\times3\times64 + 4\times64 = 1,984$ parameters
while the block in the adversary's equivalent 
CNN has $76\times64 + 4\times64 = 5,120$ parameters to be trained. Here, 
$4\times64$ learnable parameters are due to the batch normalization (\textit{bn}) layer. There are $76$ kernels in the \textit{conv'} layer and this number can be configured via the obfuscation ratio, which is set to $r=1.2$ in our example ($76 = 64\times1.2$).
In fact, even for $r=1$, the minimum allowed value, the adversary's CNN has more learnable parameters
($64\times 64 + 4\times 64 = 4,480$).

%% file: paper.bbl
\begin{thebibliography}{10}

\bibitem{hikey960}
{ HiKey 960 }.
\newblock \url{https://www.96boards.org/product/hikey960/ }.

\bibitem{android7-cts}
{Android 7.0 Compatibility Definition}.
\newblock
  \url{https://source.android.com/docs/compatibility/7.0/android-7.0-cdd#9_11_keys_and_credentials}.

\bibitem{faceapp}
{AppLock Face/Voice Recognition - Apps on Google Play}.
\newblock
  \url{https://play.google.com/store/apps/details?id=com.sensory.tsapplock\&hl=en_US\&gl=US}.

\bibitem{armcomputelib}
{Arm Compute Library}.
\newblock \url{https://www.arm.com/why-arm/technologies/compute-library}.

\bibitem{armtrustzone}
{Arm TrustZone}.
\newblock \url{https://developer.arm.com/ip-products/security-ip/trustzone}.

\bibitem{asylo}
{Asylo, An open and flexible framework for enclave applications}.
\newblock \url{https://asylo.dev/}.

\bibitem{cardapp}
{Card Scanner - Apps on Google Play}.
\newblock
  \url{https://play.google.com/store/apps/details?id=com.zoho.android.cardscanner&hl=en_US&gl=US}.

\bibitem{cifar100}
The cifar-100 dataset.
\newblock \url{https://www.cs.toronto.edu/~kriz/cifar.html}.

\bibitem{mobilegpu}
{Coral: An ecosystem for local AI}.
\newblock
  \url{https://blog.tensorflow.org/2019/01/tensorflow-lite-now-faster-with-mobile.html}.

\bibitem{edgetpu}
{Coral: An ecosystem for local AI}.
\newblock \url{https://coral.ai/about-coral/}.

\bibitem{eigen}
{Eigen}.
\newblock \url{http://eigen.tuxfamily.org/index.php?title=Main_Page}.

\bibitem{fritzai}
{Fritz AI: Model Protection - Secure your Intellectual Property from being
  tampered-with or stolen}.
\newblock \url{https://www.fritz.ai/features/model-protection.html}.

\bibitem{translateapp}
{Google Translate - Apps on Google Play}.
\newblock
  \url{https://play.google.com/store/apps/details?id=com.google.android.apps.translate&hl=en_US&gl=US}.

\bibitem{Minivgg_Code}
Minivgg pre-trained model.
\newblock
  \url{https://pyimagesearch.com/2021/05/22/minivggnet-going-deeper-with-cnns/}.

\bibitem{Mobilenets_Code}
Mobilenets pre-trained model.
\newblock \url{https://keras.io/api/applications/}.

\bibitem{Pow_repo}
Model-extraction-iclr.
\newblock \url{https://github.com/cleverhans-lab/model-extraction-iclr}.

\bibitem{optee-aosp}
{OP-TEE AOSP support}.
\newblock \url{https://optee.readthedocs.io/en/latest/building/aosp/aosp.html}.

\bibitem{optee-issue1}
Optee-os issue : Keep the context of the last session.
\newblock
  \url{https://github.com/OP-TEE/optee_os/pull/4891/commits/0fa9b4efadf8ae7a48f87184660d6b6f8e56749d}.

\bibitem{optee-issue2}
Optee-os issue :questions about memory management.
\newblock \url{https://github.com/OP-TEE/optee_os/issues/5042}.

\bibitem{Resnet404_Code}
Resnet-404 pre-trained model.
\newblock
  \url{https://github.com/wikibook/keras/blob/master/chapter2-deep-networks/resnet-cifar10-2.2.1.py}.

\bibitem{Resnet44_Code}
Resnet-44 pre-trained model.
\newblock
  \url{https://github.com/wikibook/keras/blob/master/chapter2-deep-networks/resnet-cifar10-2.2.1.py}.

\bibitem{shadownetgithub}
{ShadowNet Repo}.
\newblock \url{https://github.com/RiS3-Lab/ShadowNet}.

\bibitem{tflite}
Tensorflow lite.
\newblock \url{https://www.tensorflow.org/lite}.

\bibitem{ocrapp}
{Text Scanner [OCR] - Apps on Google Play}.
\newblock
  \url{https://play.google.com/store/apps/details?id=com.peace.TextScanner&hl=en_US&gl=US}.

\bibitem{tfencrypted}
{TF Encrypted}.
\newblock \url{https://github.com/tf-encrypted/tf-encrypted}.

\bibitem{tfdemo}
{TF Lite Android Image Classifier App Example}.
\newblock
  \url{https://github.com/tensorflow/tensorflow/tree/r2.2/tensorflow/lite/java/demo}.

\bibitem{tftrusted}
{TF Trusted}.
\newblock \url{https://github.com/dropoutlabs/tf-trusted}.

\bibitem{armtrustzone-a}
{TrustZone for Cortex-A}.
\newblock
  \url{https://www.arm.com/why-arm/technologies/trustzone-for-cortex-a}.

\bibitem{armtrustzone-m}
{TrustZone for Cortex-M}.
\newblock
  \url{https://www.arm.com/why-arm/technologies/trustzone-for-cortex-m}.

\bibitem{tee}
{Wiki: Trusted Execution Environment}.
\newblock \url{ https://en.wikipedia.org/wiki/Trusted_execution_environment}.

\bibitem{YOLOv4-tiny_Code}
Yolov4-tiny pre-trained model.
\newblock \url{https://github.com/bubbliiiing/yolov4-tiny-keras}.

\bibitem{bayerl2020offline}
Sebastian~P Bayerl, Tommaso Frassetto, Patrick Jauernig, Korbinian Riedhammer,
  Ahmad-Reza Sadeghi, Thomas Schneider, Emmanuel Stapf, and Christian Weinert.
\newblock Offline model guard: Secure and private ml on mobile devices.
\newblock {\em DATE 2020}, 2020.

\bibitem{MixMatch}
David Berthelot, Nicholas Carlini, Ian Goodfellow, Nicolas Papernot, Avital
  Oliver, and Colin~A Raffel.
\newblock Mixmatch: A holistic approach to semi-supervised learning.
\newblock In H.~Wallach, H.~Larochelle, A.~Beygelzimer, F.~d\textquotesingle
  Alch\'{e}-Buc, E.~Fox, and R.~Garnett, editors, {\em Advances in Neural
  Information Processing Systems}, volume~32. Curran Associates, Inc., 2019.

\bibitem{YOLOPaper}
Alexey Bochkovskiy, Chien{-}Yao Wang, and Hong{-}Yuan~Mark Liao.
\newblock Yolov4: Optimal speed and accuracy of object detection.
\newblock {\em CoRR}, abs/2004.10934, 2020.

\bibitem{brasser2019sanctuary}
Ferdinand Brasser, David Gens, Patrick Jauernig, Ahmad-Reza Sadeghi, and
  Emmanuel Stapf.
\newblock Sanctuary: Arming trustzone with user-space enclaves.
\newblock In {\em NDSS}, 2019.

\bibitem{livenessapp}
Cristiano Breuel.
\newblock {Implementing Liveness Detection with Google ML Kit}.
\newblock
  \url{https://towardsdatascience.com/implementing-liveness-detection-with-google-ml-kit-5e8c9f6dba45}.

\bibitem{chen2018securenets}
Xuhui Chen, Jinlong Ji, Lixing Yu, Changqing Luo, and Pan Li.
\newblock Securenets: Secure inference of deep neural networks on an untrusted
  cloud.
\newblock In {\em Asian Conference on Machine Learning}, pages 646--661, 2018.

\bibitem{pow2022iclr}
Adam Dziedzic, Muhammad~Ahmad Kaleem, Yu~Shen Lu, and Nicolas Papernot.
\newblock Increasing the cost of model extraction with calibrated proof of
  work.
\newblock In {\em ICLR (International Conference on Learning Representations)
  [SPOTLIGTH]}, 2022.

\bibitem{voc}
Mark Everingham, Luc Van~Gool, Christopher~KI Williams, John Winn, and Andrew
  Zisserman.
\newblock The pascal visual object classes (voc) challenge.
\newblock volume~88, pages 303--338. Springer, 2010.

\bibitem{ghodsi2017safetynets}
Zahra Ghodsi, Tianyu Gu, and Siddharth Garg.
\newblock Safetynets: Verifiable execution of deep neural networks on an
  untrusted cloud.
\newblock In {\em Advances in Neural Information Processing Systems}, pages
  4672--4681, 2017.

\bibitem{gilad2016cryptonets}
Ran Gilad-Bachrach, Nathan Dowlin, Kim Laine, Kristin Lauter, Michael Naehrig,
  and John Wernsing.
\newblock Cryptonets: Applying neural networks to encrypted data with high
  throughput and accuracy.
\newblock In {\em International Conference on Machine Learning}, pages
  201--210, 2016.

\bibitem{gu2018securing}
Zhongshu Gu, Heqing Huang, Jialong Zhang, Dong Su, Ankita Lamba, Dimitrios
  Pendarakis, and Ian Molloy.
\newblock Yerbabuena: Securing deep learning inference data via enclave-based
  ternary model partitioning.
\newblock {\em arXiv preprint arXiv:1807.00969}, 2018.

\bibitem{Gupta2015}
Suyog Gupta, Ankur Agrawal, Kailash Gopalakrishnan, and Pritish Narayanan.
\newblock Deep learning with limited numerical precision.
\newblock In {\em Proceedings of the 32nd International Conference on
  International Conference on Machine Learning - Volume 37}, ICML'15, page
  1737–1746. JMLR.org, 2015.

\bibitem{hanzlik2018mlcapsule}
Lucjan Hanzlik, Yang Zhang, Kathrin Grosse, Ahmed Salem, Max Augustin, Michael
  Backes, and Mario Fritz.
\newblock Mlcapsule: Guarded offline deployment of machine learning as a
  service.
\newblock {\em arXiv preprint arXiv:1808.00590}, 2018.

\bibitem{he2016deep}
Kaiming He, Xiangyu Zhang, Shaoqing Ren, and Jian Sun.
\newblock Deep residual learning for image recognition.
\newblock In {\em Proceedings of the IEEE conference on computer vision and
  pattern recognition}, pages 770--778, 2016.

\bibitem{howard2017mobilenets}
Andrew~G Howard, Menglong Zhu, Bo~Chen, Dmitry Kalenichenko, Weijun Wang,
  Tobias Weyand, Marco Andreetto, and Hartwig Adam.
\newblock Mobilenets: Efficient convolutional neural networks for mobile vision
  applications.
\newblock {\em arXiv preprint arXiv:1704.04861}, 2017.

\bibitem{jagielski2019high}
Matthew Jagielski, Nicholas Carlini, David Berthelot, Alex Kurakin, and Nicolas
  Papernot.
\newblock High-fidelity extraction of neural network models.
\newblock {\em arXiv preprint arXiv:1909.01838}, 2019.

\bibitem{ji2019microtee}
Dongxu Ji, Qianying Zhang, Shijun Zhao, Zhiping Shi, and Yong Guan.
\newblock Microtee: designing tee os based on the microkernel architecture.
\newblock In {\em 2019 18th IEEE International Conference On Trust, Security
  And Privacy In Computing And Communications/13th IEEE International
  Conference On Big Data Science And Engineering (TrustCom/BigDataSE)}, pages
  26--33. IEEE, 2019.

\bibitem{jiang2018secure}
Xiaoqian Jiang, Miran Kim, Kristin Lauter, and Yongsoo Song.
\newblock Secure outsourced matrix computation and application to neural
  networks.
\newblock In {\em Proceedings of the 2018 ACM SIGSAC Conference on Computer and
  Communications Security}, pages 1209--1222, 2018.

\bibitem{juvekar2018gazelle}
Chiraag Juvekar, Vinod Vaikuntanathan, and Anantha Chandrakasan.
\newblock $\{$GAZELLE$\}$: A low latency framework for secure neural network
  inference.
\newblock In {\em 27th $\{$USENIX$\}$ Security Symposium ($\{$USENIX$\}$
  Security 18)}, pages 1651--1669, 2018.

\bibitem{KatzLindell}
Jonathan Katz and Yehuda Lindell.
\newblock {\em Introduction to Modern Cryptography, Second Edition}.
\newblock Chapman \& Hall/CRC, 2nd edition, 2014.

\bibitem{cifar10}
Alex Krizhevsky, Vinod Nair, and Geoffrey Hinton.
\newblock Cifar-10 (canadian institute for advanced research).

\bibitem{kunkel2019tensorscone}
Roland Kunkel, Do~Le Quoc, Franz Gregor, Sergei Arnautov, Pramod Bhatotia, and
  Christof Fetzer.
\newblock Tensorscone: A secure tensorflow framework using intel sgx.
\newblock {\em arXiv preprint arXiv:1902.04413}, 2019.

\bibitem{lentz2018secloak}
Matthew Lentz, Rijurekha Sen, Peter Druschel, and Bobby Bhattacharjee.
\newblock Secloak: Arm trustzone-based mobile peripheral control.
\newblock In {\em Proceedings of the 16th Annual International Conference on
  Mobile Systems, Applications, and Services}, pages 1--13, 2018.

\bibitem{optee}
Linaro.
\newblock {Open Portable Trusted Execution Environment}.
\newblock \url{https://www.op-tee.org/}.

\bibitem{darknetz}
Fan Mo, Ali~Shahin Shamsabadi, Kleomenis Katevas, Soteris Demetriou, Ilias
  Leontiadis, Andrea Cavallaro, and Hamed Haddadi.
\newblock Darknetz: Towards model privacy at the edge using trusted execution
  environments.
\newblock In {\em ACM MobiSys 2020}.

\bibitem{Knockoff}
Tribhuvanesh Orekondy, Bernt Schiele, and Mario Fritz.
\newblock Knockoff nets: Stealing functionality of black-box models.
\newblock In {\em 2019 IEEE/CVF Conference on Computer Vision and Pattern
  Recognition (CVPR)}, pages 4949--4958, 2019.

\bibitem{o2015introduction}
Keiron O'Shea and Ryan Nash.
\newblock An introduction to convolutional neural networks.
\newblock {\em arXiv preprint arXiv:1511.08458}, 2015.

\bibitem{papernot2017practical}
Nicolas Papernot, Patrick McDaniel, Ian Goodfellow, Somesh Jha, Z~Berkay Celik,
  and Ananthram Swami.
\newblock Practical black-box attacks against machine learning.
\newblock In {\em Proceedings of the 2017 ACM on Asia conference on computer
  and communications security}, pages 506--519, 2017.

\bibitem{darknet13}
Joseph Redmon.
\newblock Darknet: Open source neural networks in c.
\newblock \url{http://pjreddie.com/darknet/}, 2013--2016.

\bibitem{DBLP:journals/corr/RedmonDGF15}
Joseph Redmon, Santosh~Kumar Divvala, Ross~B. Girshick, and Ali Farhadi.
\newblock You only look once: Unified, real-time object detection.
\newblock {\em CoRR}, abs/1506.02640, 2015.

\bibitem{riazi2018chameleon}
M~Sadegh Riazi, Christian Weinert, Oleksandr Tkachenko, Ebrahim~M Songhori,
  Thomas Schneider, and Farinaz Koushanfar.
\newblock Chameleon: A hybrid secure computation framework for machine learning
  applications.
\newblock In {\em Proceedings of the 2018 on Asia Conference on Computer and
  Communications Security}, pages 707--721, 2018.

\bibitem{doi:10.1080/0161-119691885040}
Frank Rubin.
\newblock One-time pad cryptography.
\newblock {\em Cryptologia}, 20(4):359--364, 1996.

\bibitem{ILSVRC15}
Olga Russakovsky, Jia Deng, Hao Su, Jonathan Krause, Sanjeev Satheesh, Sean Ma,
  Zhiheng Huang, Andrej Karpathy, Aditya Khosla, Michael Bernstein,
  Alexander~C. Berg, and Li~Fei-Fei.
\newblock {ImageNet Large Scale Visual Recognition Challenge}.
\newblock {\em International Journal of Computer Vision (IJCV)},
  115(3):211--252, 2015.

\bibitem{salinas2016efficient}
Sergio Salinas, Changqing Luo, Weixian Liao, and Pan Li.
\newblock Efficient secure outsourcing of large-scale quadratic programs.
\newblock In {\em Proceedings of the 11th ACM on Asia Conference on Computer
  and Communications Security}, pages 281--292, 2016.

\bibitem{simonyan2014very}
Karen Simonyan and Andrew Zisserman.
\newblock Very deep convolutional networks for large-scale image recognition.
\newblock {\em arXiv preprint arXiv:1409.1556}, 2014.

\bibitem{sun2020mind}
Zhichuang Sun, Ruimin Sun, and Long Lu.
\newblock Mind your weight (s): A large-scale study on insufficient machine
  learning model protection in mobile apps.
\newblock {\em arXiv preprint arXiv:2002.07687}, 2020.

\bibitem{szegedy2015going}
Christian Szegedy, Wei Liu, Yangqing Jia, Pierre Sermanet, Scott Reed, Dragomir
  Anguelov, Dumitru Erhan, Vincent Vanhoucke, and Andrew Rabinovich.
\newblock Going deeper with convolutions.
\newblock In {\em Proceedings of the IEEE conference on computer vision and
  pattern recognition}, pages 1--9, 2015.

\bibitem{DBLP:journals/corr/abs-2112-02439}
Tianxiang Tan and Guohong Cao.
\newblock Deep learning on mobile devices through neural processing units and
  edge computing.
\newblock {\em CoRR}, abs/2112.02439, 2021.

\bibitem{tramer2018slalom}
Florian Tramer and Dan Boneh.
\newblock Slalom: Fast, verifiable and private execution of neural networks in
  trusted hardware.
\newblock {\em arXiv preprint arXiv:1806.03287}, 2018.

\bibitem{tramer2016stealing}
Florian Tram{\`e}r, Fan Zhang, Ari Juels, Michael~K Reiter, and Thomas
  Ristenpart.
\newblock Stealing machine learning models via prediction apis.
\newblock In {\em 25th $\{$USENIX$\}$ Security Symposium ($\{$USENIX$\}$
  Security 16)}, pages 601--618, 2016.

\bibitem{volos2018graviton}
Stavros Volos, Kapil Vaswani, and Rodrigo Bruno.
\newblock Graviton: Trusted execution environments on gpus.
\newblock In {\em 13th $\{$USENIX$\}$ Symposium on Operating Systems Design and
  Implementation ($\{$OSDI$\}$ 18)}, pages 681--696, 2018.

\bibitem{xu2019first}
Mengwei Xu, Jiawei Liu, Yuanqiang Liu, Felix~Xiaozhu Lin, Yunxin Liu, and
  Xuanzhe Liu.
\newblock A first look at deep learning apps on smartphones.
\newblock In {\em The World Wide Web Conference}, pages 2125--2136, 2019.

\bibitem{yosinski2014transferable}
Jason Yosinski, Jeff Clune, Yoshua Bengio, and Hod Lipson.
\newblock How transferable are features in deep neural networks?
\newblock In {\em Advances in neural information processing systems}, pages
  3320--3328, 2014.

\bibitem{zhang2020decomposable}
Yi~Zhang, Jiun-Hao Liu, Chih-Yu Wang, and Hung-Yu Wei.
\newblock Decomposable intelligence on cloud-edge iot framework for live video
  analytics.
\newblock {\em IEEE Internet of Things Journal}, 7(9):8860--8873, 2020.

\end{thebibliography}
